\documentclass{article}

\usepackage{basic}

\title{Mechanism Design for Congested Facility Location}
\author{Cheng Peng\thanks{Chinese University of Hong Kong. Email: chengpeng9660@gmail.com} \hspace{3em}
        Houyu Zhou\thanks{UNSW Sydney. Email: houyu.zhou97@gmail.com}}
\date{}

\begin{document}
\maketitle
\begin{abstract}
 This paper investigates mechanism design for congested facility location, where agents are partitioned into groups based on conflicting interests (e.g., competition for booking a basketball court in the gymnasium), and each agent’s cost increases when the facility is closer to their competitors. We analyze three types of misreporting: location only, group membership only, and both. To minimize social cost, we propose a strategyproof mechanism that achieves optimality for location-only misreporting. For group membership and combined misreporting, we demonstrate that the median mechanism yields tight approximation bounds. For minimizing maximum cost, we introduce novel strategyproof mechanisms for location-only and group-membership-only misreporting, while employing the leftmost mechanism for combined misreporting. We prove that all proposed mechanisms achieve near-tight bounds.
\end{abstract}
\thispagestyle{empty}

\vspace{15pt}

\hrule

\vspace{5pt}
{
	\setlength\columnsep{35pt}
	\setcounter{tocdepth}{2}
	\renewcommand\contentsname{Contents}
		{\small\tableofcontents}
}

\vspace{11pt}
\hrule

\newpage
\pagenumbering{arabic}

%

\section{Introduction}\label{sec:intro}

Classical facility location mechanism design addresses the strategic placement of public resources to serve agents, typically modeled as points in a metric space with preferences based solely on their distance to the facility \citep{procaccia2009approximate, chan2021survey}. However, these models overlook congestion effects, which arise from agents with conflicting interests sharing limited resources and are critical in many real-world scenarios like urban planning, network optimization, and collective decision-making. In real-world settings, agents often belong to groups based on shared traits, such as affiliations or resource needs, leading to intra-group competition \citep{simon1952formal, cioffi2014introduction}. For example, a basketball team may desire a gym close to their practice area for convenience but far from rival teams to avoid scheduling conflicts. Meanwhile, they remain indifferent to the preferences of runners or tennis players, whose interests do not overlap. Similarly, in urban environments, drivers grouped by vehicle type (e.g., electric or internal combustion engine) compete for limited parking, creating congestion within groups without affecting other groups.

To model these dynamics, we propose a mechanism design framework for congested facility location, where agents are divided into groups based on conflicting interests, and each agent’s cost includes their distance to the facility plus a penalty proportional to their group’s total benefit from the facility’s proximity, scaled by a group-specific coefficient. The social planner aims to place the facility to maximize either utilitarian welfare (minimizing total individual cost) or egalitarian welfare (minimizing maximum individual cost), while ensuring truthful reporting of locations and group memberships. This leads to a key question:

\emph{How can a social planner design a mechanism to optimally locate a facility, given agents’ reported locations and group memberships, while incentivizing truthful reporting and achieving near-optimal welfare under congestion?}

This question extends classical facility location mechanism design, where agents seek to minimize their distance to the facility \citep{procaccia2009approximate}. Recent studies have incorporated externalities to capture agent interactions. For instance, \citet{li2019facility} explored pairwise externalities, while \citet{zhou22COCOON} modeled a single group with positive externalities, and \citet{wang24externalities} extended this to multiple groups with cooperative interactions. However, these works do not address competitive intra-group congestion, which is critical in scenarios like resource allocation in smart cities or distributed systems. Our work fills this gap by introducing a novel model with negative group externalities and designing strategyproof mechanisms to address congestion effects.

\subsection{Our Contribution}
We study mechanism design in facility location problems with multiple groups, where each agent’s cost combines their distance to the facility with a negative externality from their group members, proportional to the group’s total benefit from distance to the facility (see \Cref{sec:pre} for formal definition). We consider two objectives--minimizing social cost (utilitarian welfare) and maximum cost (egalitarian welfare)--and analyze three types of misreporting: location only, group membership only, and both. Our contributions are summarized in \Cref{tab:sum}. In particular,
\begin{table}[!htb]
\centering
\begin{tabular}{l|lll}
\toprule
\multirow{2}{*}{Objective}    & \multicolumn{3}{l}{Restrict the Power of Misreport to}                     \\ 
                              & \multicolumn{1}{l|}{Location only} & \multicolumn{1}{l|}{Group only} & Both \\ \hline
\multirow{2}{*}{Social Cost}  & \multicolumn{1}{l|}{UB: 1}          & \multicolumn{1}{l|}{UB: 2}       & UB: 2 \\ 
                              & \multicolumn{1}{l|}{LB: 1}          & \multicolumn{1}{l|}{LB: 2$^*$}   & LB: 2 \\ \hline
\multirow{2}{*}{Maximum Cost} & \multicolumn{1}{l|}{UB: 2.125}      & \multicolumn{1}{l|}{UB: 1.774$^\dagger$} & UB: 3 \\ 
                              & \multicolumn{1}{l|}{LB: 2}          & \multicolumn{1}{l|}{LB: 1.094$^\ddagger$} & LB: 2 \\
\bottomrule
\end{tabular}
\caption{Summary of Results. $^*:$ the lower bound of all strategyproof mechanisms that output an agent location. $^\dagger: (29+20\sqrt{10})/54$. $^\ddagger: 2(1+2\sqrt{2})/7$.}
\label{tab:sum}
\vspace{-1em}
\end{table}
\begin{itemize}
    \item \textbf{Minimizing Social Cost:}
    \begin{itemize}
        \item \emph{Location-only misreporting}: By analyzing each agent’s contribution to the social cost, we characterize all optimal facility locations and propose a strategyproof mechanism, Res-M, which selects an optimal location, achieving an approximation ratio of 1.
        \item \emph{Group-membership-only misreporting}: We show that the strategyproof mechanism Med-M, which places the facility at the median agent location, achieves an approximation ratio of 2. We also prove a lower bound of 2 for all strategyproof mechanisms that select an agent’s location, indicating Med-M’s optimality in this class.
        \item \emph{Unrestricted misreporting}: Reusing Med-M, we demonstrate that any deterministic strategyproof mechanism has an approximation ratio of at least 2, constructed via a series of profile sets, confirming that Med-M is the best possible for this setting.
    \end{itemize}
    \item \textbf{Minimizing Maximum Cost:}
    \begin{itemize}
        \item \emph{Location-only misreporting}: We propose a strategyproof mechanism, LoF-M, which selects the facility location based on the leftmost and rightmost agents in each group, achieving an approximation ratio of 17/8. Since classical facility location mechanism design \cite{procaccia2009approximate} is a special case of our setting, the lower bound of 2 applies, making LoF-M nearly optimal.
        \item \emph{Group-membership-only misreporting}: We introduce Mid-M, a strategyproof mechanism that places the facility at the midpoint between the leftmost and rightmost agent locations, achieving an approximation ratio of \((29 + 20\sqrt{10})/54 \approx 1.774\). We also establish a lower bound of \(2(1 + 2\sqrt{2})/7 \approx 1.094\) for this setting.
        \item \emph{Unrestricted misreporting}: Using the Left-M mechanism, we achieve an approximation ratio of 3 and prove a lower bound of 2, highlighting the challenges of unrestricted misreporting.
    \end{itemize}
\end{itemize}

\subsection{Related Work}

The study of mechanism design in facility location problems was pioneered by \citet{procaccia2009approximate}, who introduced approximate mechanism design for settings where agents aim to minimize their distance to the facility. \citet{li2019facility} extended this to include pairwise externalities, though without group structures or competitive dynamics. \citet{zhou22COCOON} considered a single group with positive externalities, separating agents into participants and non-participants, while \citet{wang24externalities} generalized this to multiple groups with cooperative interactions. In contrast, our focus on negative group externalities introduces a novel competitive dynamic not addressed in prior work, with direct applications to AI-driven resource allocation.

Other works have explored mechanism design in group-based facility location problems. \citet{filos2021approximate,filos2023approximate} proposed a distributed process where groups select representative locations, and the mechanism chooses from these without knowing individual agent locations. \citet{DBLP:conf/ijcai/ZhouLC22} and \citet{DBLP:conf/aaai/ZhouCL24} studied group-fair and altruistic objectives, respectively, but neither considered negative externalities. Related fields, such as cake cutting \cite{DBLP:conf/ijcai/BranzeiPZ13,DBLP:conf/ijcai/LiZZ15}, fair allocation \cite{DBLP:conf/pricai/MishraPG22,gan2023your}, and matching \cite{sasaki1996two,pycia2023matching}, also explore externalities but do not address our specific competitive group dynamics. For a comprehensive overview, see \cite{chan2021survey}.


\section{Preliminaries}~\label{sec:pre}

Let \( N = \{1, 2, \ldots, n\} \) be a set of agents located on a normalized closed interval \( I = [0,1] \). Each agent \( i \in N \) has a location \( x_i \in I \) and belongs to a group \( g_i \in [m] \), where \( [m] = \{1, 2, \ldots, m\} \). The set of agents in group \( j \) is denoted \( G_j = \{ i \in N : g_i = j \} \), such that \( \bigcup_{j \in [m]} G_j = N \) and \( G_{j_1} \cap G_{j_2} = \emptyset \) for all \( j_1 \neq j_2 \). We denote the profile of agent \( i \) as \( r_i = (x_i, g_i) \), and the profile set of all agents as \( \mathbf{r} = \{r_1, \ldots, r_n\} \). A deterministic mechanism is a function \( f \) that maps a profile set \( \mathbf{r} \) to a facility location \( y \in I \). The distance between points \( a, b \in I \) is defined as \( d(a, b) = |a - b| \).

The cost of agent \( i \in G_{g_i} \) for a facility located at \( y \in I \) is defined as:
\[
c_i(y, \mathbf{r}) = d(y, x_i) + \alpha_{g_i} \sum_{k \in G_{g_i} \setminus \{i\}} (1 - d(y, x_k)),
\]
where \( \alpha_{g_i} \geq 0 \) is the externality factor for group \( g_i \). The first term, \( d(y, x_i) \), represents the agent’s cost of accessing the facility (e.g., travel time), while the second term captures a negative externality from intra-group competition (e.g., for limited resources like gym slots or server bandwidth).

To ensure realistic agent preferences, we introduce the \emph{Consistency of Aversion} property. Agents desire to be close to a facility, such as a community center, but face increased costs due to competition with group members. As the facility moves farther from an agent, particularly in the direction where all competitors are located, their aversion---manifested as a higher cost---should consistently grow, reflecting greater inconvenience without relief from group rivalry. For example, consider three agents in group \( G_j \) at locations \( x_1 = 0.1 \), \( x_2 = 1 \), \( x_3 = 1 \). For agent 1 at \( x_1 = 0.1 \), with competitors at \( x_2 = x_3 = 1 \), as the facility moves from \( y = 0.1 \) to \( y = 1 \), the cost increases due to the growing distance \( d(y, 0.1) \) (from 0 to 0.9) and heightened competition from group members. Beyond \( y = 1 \), the cost should continue to rise, as agent 1 becomes increasingly inconvenienced without further competitive benefit to group members, as shown in \Cref{fig:aversion}.

\begin{figure}[!htb]
\centering
\resizebox{0.6\linewidth}{!}{
\begin{tikzpicture}
\begin{axis}[
    axis lines=middle,
    xlabel={Facility Location $y$},
    ylabel={Cost for \textcolor{white!50!blue}{Agent 1}},
    domain=0:1.5,
    ymin=0, ymax=2.2,
    xmin=0, xmax=1.5,
    width=\linewidth,
    height=0.6\linewidth,
    grid=major,
    legend style={at={(0.95,0.15)}, anchor=south east, font=\small},
    samples=100
]
\addplot[black!50!green, thick, domain=0:1.5] {x <= 0.1 ? 0.1 - 0.2*x : (x <= 1 ? -0.1 + 1.8*x : 1.5 + 0.2*x)};
\addlegendentry{$\alpha=0.4$}
\addplot[black!50!red, thick, domain=0:1.5] {x <= 0.1 ? 0.1 + 0.2*x : (x <= 1 ? -0.1 + 2.2*x : 2.3 - 0.2*x)};
\addlegendentry{$\alpha=0.6$}
\addplot[only marks, mark=*, mark size=3, white!50!blue] coordinates {(0.1,0)};
\addplot[only marks, mark=*, mark size=3, black] coordinates {(0.98,0) (1.02,0)};
\end{axis}
\end{tikzpicture}
}
\caption{Cost for \textcolor{white!50!blue}{agent 1} at \( x_1 = 0.1 \), with competitors at \( x_2 = x_3 = 1 \). For \( \alpha_j = 0.4 \), the cost is non-decreasing for \( y \geq 0.1 \), \textcolor{black!50!green}{satisfying} \emph{Consistency of Aversion}; for \( \alpha_j = 0.6 \), the cost decreases for $y\ge 1$, \textcolor{black!50!red}{violating} the property.}
\label{fig:aversion}
\end{figure}
Formally, the \emph{Consistency of Aversion} property requires that when all competitors are on one side of agent \( i \) (i.e., \( x_k \geq x_i \) or \( x_k \leq x_i \) for all \( k \in G_{g_i} \setminus \{i\} \)), the cost \( c_i(y, \mathbf{r}) \) is non-decreasing as \( y \) moves further in that direction.

\begin{proposition}~\label{pro:aversion}
The cost function \( c_i(y, \mathbf{r}) = d(y, x_i) + \alpha_{g_i} \sum_{k \in G_{g_i} \setminus \{i\}} (1 - d(y, x_k)) \) satisfies the \emph{Consistency of Aversion} property if and only if \( \alpha_j \leq 1/(|G_j|-1) \) for each group \( j \in [m] \).
\end{proposition} 

We focus on the agent cost function satisfying \emph{Consistency of Aversion}. In addition, each agent is rational and tries to minimize their cost. Our goal is to design mechanisms that are strategyproof while (approximately) optimizing an objective function.

\begin{definition}
    A mechanism \( f \) is \textbf{strategyproof} (SP) if an agent can never benefit by reporting a false profile, regardless of the strategies of other agents. Formally, for any true profile set \( \mathbf{r} = \{r_1, \ldots, r_n\} \) and reported profile set \( \mathbf{r'} = \{r'_1, \ldots, r'_n\} \), we have \( c_i(f(r_i, \mathbf{r'}_{-i}), \mathbf{r}) \leq c_i(f(\mathbf{r'}), \mathbf{r}) \), where \( \mathbf{r'}_{-i} \) is the collection of reported profiles of all agents except agent \( i \).
\end{definition}

We can further discuss the three cases: 1. Restrict the power of misreport to locations only; 2. Restrict the power of misreport to group memberships only; 3. Do not restrict the power of misreport. We consider two cost objectives, minimizing the \textbf{s}ocial \textbf{c}ost, and the \textbf{m}aximum \textbf{c}ost:
\begin{equation*}
    \soc(y, \mathbf{r})=\sum_{i\in N}c_i(y, \mathbf{r})\text{ and }    \mc(y, \mathbf{r})=\max_{i\in N}c_i(y, \mathbf{r}).
\end{equation*}
We measure the performance of $f$ by approximation ratio.\footnote{In facility location mechanism design on the interval \( I = [0,1] \), locating the facility at \( y = 1/2 \) often guarantees a 2-approximation ratio for utilitarian/egalitarian welfare, driven by simple utility functions (\( 1 - d(y, x_i) \)). In our setting, this mechanism yields unbounded approximation ratios for both objectives, demonstrating that the interval setting alone does not mitigate the problem’s complexity.}
    
\begin{definition}
    We say a mechanism $f$ has an \textbf{approximation ratio} of $\rho$ for a certain objective if there exists a number $\rho$ such that for any profile set $\mathbf{r}$, the objective value achieved by $f$ is within $\rho$ times the objective value achieved by the optimal facility location.
\end{definition}

When $\forall j\in [m], \alpha_j=0$, our setting can be viewed as the classical facility location problems~\citep{procaccia2009approximate}.  In this special setting, locating the facility at the median agent location and the leftmost agent location are the best strategyproof mechanisms for the social cost and the maximum cost, respectively.

\begin{mechanism}[Median-point Mechanism (Med-M)]
    Given any profile set $\mathbf{r}$, without loss of generality we assume that $x_1\leq\dots\leq x_n$, locate the facility at $y=x_{\lceil \frac{n}{2} \rceil}$.
\end{mechanism} 

\begin{mechanism}[Leftmost-point Mechanism (Left-M)]
    Given any profile set $\mathbf{r}$, without loss of generality we assume that $x_1\leq\dots\leq x_n$, locate the facility at $y=x_1$.
\end{mechanism}

Hence, a natural question is whether Med-M and Left-M remain optimal in this more general setting. In the following sections, we first investigate minimizing the social cost and then focus on minimizing the maximum cost. Some proofs are deferred to the Appendix due to space constraints.

\section{Social Cost}~\label{sec:social}
To minimize the social cost, we first consider the performance of Med-M. Though it is not hard to show the strategyproofness of Med-M, the analysis of the approximation ratio in this new setting is obviously non-trivial. After showing that, we propose a mechanism that has a smaller approximation ratio.

\begin{theorem}~\label{thm:med-social}
    Med-M is strategyproof when both locations and group memberships are private, and has an approximation ratio of $2$.
\end{theorem}

\begin{proof}[proof sketch]
    For the strategyproofness, we can discuss by cases: agent on the left of $y$, at $y$, and on the right of $y$. Next, we consider its approximation ratio.
    
    Given any profile set $\mathbf{r}$, without loss of generality we assume that the optimal facility location $y^*>y$ ($y^*=y$ implies the approximation ratio 1). Note that the social cost achieved by $y$ can be reformulated as
      \begin{align*}
          & \soc(y, \mathbf{r}) = \sum_{i\in N}\left (d(x_i,y)+\alpha_{g_i}\sum_{k\in G_{g_i};k\neq i}(1-d(y,x_k))\right )\\
          &= \sum_{i\in N}\Big( (1-\alpha_{g_i}(|G_{g_i}|-1))d(y,x_i)+\alpha_{g_i}(|G_{g_i}|-1)\Big ).
      \end{align*}
      Then we have
      \begin{align*}
         \soc(y,\mathbf{r})-\soc(y^*,\mathbf{r}) = \sum_{i\in N} w_i (d(y, x_i)-d(y^*, x_i)),
      \end{align*}
     where $w_i = 1-\alpha_{g_i}(|G_{g_i}|-1)$.
     For agents at or on the right of $y$, we have $d(y, x_i)-d(y^*, x_i) \le y^*-y$, where the equality holds when agent $i$ is at or on the right of $y^*$. For agents on the left of $y$, we have $d(y, x_i)-d(y^*, x_i) = y-y^*$. Hence, we futher have
     \begin{align*}
         \soc(y,\mathbf{r})-\soc(y^*,\mathbf{r}) \le (\sum_{x_i> y} w_i - \sum_{x_i\le y} w_i)(y^*-y),
     \end{align*}
    
     Moreover, the social cost incurred by $y^*$ is at least
     \begin{align*}
         \soc(y^*, \mathbf{r}) &\ge\!\! \sum_{x_i\le y} \!(w_i d(x_i, y^*) \!+\! C_i) + \sum_{x_i> y} \!(w_i d(x_i, y^*) \!+\! C_i)\\
         &\ge\! \sum_{x_i\le y} \!(w_i (y^*-y) \!+\! C_i)+\sum_{x_i> y} \!C_i,
     \end{align*}
    where $C_i = \alpha_{g_i}(|G_{g_i}|-1)$.
     Then we have the approximation ratio is at most
     \begin{align*}
        \rho \le & \frac{\sum_{x_i> y} w_i(y^*-y)+\sum_{x_i\le y} C_i + \sum_{x_i> y} C_i}{\sum_{x_i\le y} w_i (y^*-y) +  \sum_{x_i\le y} C_i + \sum_{x_i> y} C_i}\\
         \le & \frac{\sum_{x_i> y} w_i+\sum_{x_i\le y} C_i + \sum_{x_i> y} C_i}{\sum_{x_i\le y} w_i +  \sum_{x_i\le y} C_i + \sum_{x_i> y} C_i}\\
         = & \frac{\sum_{x_i> y} 1+\sum_{x_i\le y} C_i}{\sum_{x_i\le y} 1 + \sum_{x_i> y} C_i}.
     \end{align*}
    
     Let $c_{\max} = \max_{j \in [m]}C_j$ and $c_{\min} = \min_{j \in [m]}C_j$. Since $y$ is the median location, the approximation ratio is at most
     \begin{align*}
         \rho &= \frac{\sum_{x_i> y} 1+\sum_{x_i\le y} C_i}{\sum_{x_i\le y} 1 + \sum_{x_i> y} C_i} \le \frac{\frac{n}{2}+\frac{n}{2}c_{\max}}{\frac{n}{2}+\frac{n}{2}c_{\min}}=\frac{1+c_{\max}}{1+c_{\min}}.
     \end{align*}
    
     Since $\alpha_j (|G_j|-1) \in [0,1]$, we have the approximation ratio of $2$.
\end{proof}
From the proof of Theorem \ref{thm:med-social}, we can see that Med-M is optimal when all groups are identical, i.e. all groups have the same externality and size. However, when groups are diverse, Med-M achieves a larger approximation ratio. Hence, we need to leverage the group information to design mechanisms. Moreover, one can foresee that Left-M is even worse than Med-M in minimizing the social cost.

\subsection{Private Locations and Public Group Memberships}
Before showing the mechanism, let us focus on the contribution of each agent to the social cost in depth. The intuitive idea is that each agent contributes the social cost with the amount of their own cost. However, from an alternative macro perspective, the social cost contributed by each agent location can also be divided into two parts, their distance in their own cost and their distance in their group members' cost. Hence, we can alternatively regard the contribution of each agent to the social cost is 
\begin{align*}
    &d(x_i,y)+\alpha_{g_i}(|G_{g_i}|-1)(1-d(x_i,y)) \\
    =& (1-\alpha_{g_i}(|G_{g_i}|-1))d(x_i, y) + \alpha_{g_i}(|G_{g_i}|-1),
\end{align*}
which has already been used in the proof of Theorem~\ref{thm:med-social}. Let $w_i=1-\alpha_{g_i}(|G_{g_i}|-1)$, we will characterize the optimal facility location based on the above discussion. Let $L(y)$ be the set of the agents with $x_i\le y$ and $R(y)$ be the set of the  agents with $x_i>y$, respectively.
\begin{observation}
If $y_1 < y_2$, then $L(y_1) \subseteq L(y_2)$ and $R(y_2) \subseteq R(y_1)$.
\end{observation}
From the observation, we can see that the function $l(y) = \sum_{i \in L(y)} w_i$ is non-decreasing on $y$.
Analogously, the function $r(y) = \sum_{i \in R(y)} w_i$ is non-increasing on $y$. Hence, we can conclude that the following sets are continuous and bounded.
\begin{align*}
Y_1 = \Big \{y \Big | \sum_{i \in L(y)} w_i <  \sum_{i \in R(y)} w_i \Big \},\\
Y_2 = \Big \{y \Big | \sum_{i \in L(y)} w_i > \sum_{i \in R(y)} w_i \Big \}.
\end{align*}
Moreover, we can see that $Y_1 \cap Y_2 = \emptyset$. Let $y_l = \sup Y_1$ and $y_r = \inf Y_2$. Then we characterize the optimal facility location.

\begin{proposition}~\label{pro:opt}
Given any agent profile $\mathbf{r} = (r_1, r_2, \ldots, r_n)$, $y^*$ is the optimal facility location if and only if $y^* \in [y_l, y_r]$.
\end{proposition}

\begin{proof}
   Given any profile set $\mathbf{r}$, the social cost of putting facility at $y$ is
\begin{align*}
    &\soc(y, \mathbf{r}) =\sum_{i \in N} \left[w_i d(x_i, y) + \alpha_{g_i}(|G_{g_i}|-1) \right]. 
\end{align*}
The derivative of this $\soc(y,\mathbf{r})$ function is
\begin{equation*}
\frac{\partial \soc(y, \mathbf{r})}{\partial y} = \sum_{i \in N} \frac{\partial c_i(y, \mathbf{r})}{\partial y} = \sum_{i \in N} \left[ w_i  \frac{\partial d(x_i, y)}{\partial y} \right]
\end{equation*}
\begin{equation*}
=-\sum_{i \in L(y)}  w_i  +\sum_{i \in R(y)}w_i.
\end{equation*} 
From the definition of $y_l$ and $y_r$, we have
\begin{equation*}
  \frac{\partial \soc(y, \mathbf{r})}{\partial y}<0 (y<y_l),  
\text{ and }
  \frac{\partial \soc(y, \mathbf{r})}{\partial y}>0 (y>y_r).  
\end{equation*}
Besides, the social cost function is continuous. We can conclude the optimal facility location is between $y_l$ and $y_r$.
\end{proof}

Observe that $\sum_{i \in L(y)} w_i$ and $\sum_{i \in R(y)} w_i$  can only change at  $x_i$, for $i = 1, 2, \ldots, n$. Hence, we have the following mechanism outputting the optimal solution, which restructures the computation of the social cost as we have discussed.

\begin{mechanism}[Restructure Mechanism (Res-M)]
    Given any agent profile $\mathbf{r} = (r_1, \ldots, r_n)$, without loss of generality, we assume that $x_1 \le x_2 \le \dots \le x_n$, output the smallest $x_i^*$ such that
    \begin{equation*}
    \sum_{i \in L(x_i^*)} w_i \ge \sum_{i \in R(x_i^*)} w_i,
    \end{equation*}
    where $w_i=1-\alpha_{g_i}(|G_{g_i}|-1)$.
\end{mechanism}

Now we show the strategyproofness of Res-M.
\begin{theorem}~\label{thm:res}
Res-M is strategyproof and optimal for minimizing the social cost when locations are private and group memberships are public.
\end{theorem}
Hence, when agents only misreport their locations, the social cost can in fact be minimized using Res-M. Next, we consider the setting of private group memberships.

\subsection{Private Group Memberships}
In this subsection, Res-M is no longer strategyproof. Hence, we use Med-M in this setting. Since we have shown the strategyproofness and the approximation of Med-M in Theorem~\ref{thm:med-social}, we only need to show the lower bound to complete the results. We first show a lemma that helps us to derive the lower bounds.

\begin{lemma} \label{lem:partial group SP}
    If a mechanism is strategyproof, then for any set of agents at the same location, each individual cannot benefit if they misreport their group memberships simultaneously.
\end{lemma}
\begin{theorem}~\label{thm:lbp}
    Any deterministic mechanism which outputs the facility at an agent location has an approximation ratio of at least $2$ for minimizing the social cost when locations are public and group memberships are private.
\end{theorem}
Hence, Med-M is the best among all strategyproof mechanisms that output an agent location. We also show the lower bound of all strategyproof mechanisms in this setting, which can be found in Appendix.

When locations are private, we will show that the Med-M is the best among all strategyproof mechanisms, not just strategyproof mechanisms that output an agent location. First, we will present a lemma to help us derive the lower bound.

\begin{lemma} \label{lem:partial group SP location}
    If a mechanism is strategyproof, then for any set of agents at the same location, each individual cannot benefit if they misreport their locations simultaneously.
\end{lemma}
\begin{theorem}~\label{thm:pp2}
    Any deterministic mechanism has an approximation ratio of at least $2$ for minimizing the social cost when both locations and group memberships are private.
\end{theorem}
Next, we consider minimizing the maximum cost.  Unlike most previous work on facility location mechanism design, where analyzing the maximum cost faces less challenge than analyzing the social cost, we face the greatest challenge with this objective.

\section{Maximum Cost}~\label{sec:maximum}
Similar to the social cost, we first consider the performance of Left-M, the best mechanism for the maximum cost in the classical facility location mechanism design~\citep{procaccia2009approximate}. After showing that, we propose a mechanism that has a smaller approximation ratio.

\begin{lemma}~\label{lem:single group}
    Locating the facility at any point between the leftmost agent location and the rightmost agent location can achieve an approximation ratio of $2$, when there is only one group.
\end{lemma}
\begin{theorem}~\label{thm:5}
    Left-M is strategyproof when both locations and group memberships are private, and has an approximation ratio of 3.
\end{theorem}
\subsection{Private Locations and Public Group Memberships}

Different from the social cost, we cannot restructure the contribution of each agent since the maximum cost is incurred by a single agent. From Lemma~\ref{lem:single group} we know that the approximation ratio within a group is very small. 
Hence, the following direction of strategyproof mechanism design is to find a location such that the analysis of the approximation ratio can be degenerated to the single group in most of cases.

\begin{mechanism}[Last-leftmost or First-rightmost Mechanism (LoF-M)]
    Given any profile set $\mathbf{r} = (r_1, \ldots, r_n)$, without loss of generality, output the facility location $y$ at the leftmost location among the rightmost group leftmost agent location or the leftmost group rightmost agent location, i.e.,
    \begin{equation*}
        y = \min\{x_{lm}^*, x_{rm}^*\}
    \end{equation*}
    where $x_{lm}^* = \max_{j\in [m]}\{x^j_{lm}\}$, $x_{rm}^* = \min_{j\in [m]}\{x^j_{rm}\}$, $x^j_{lm}$ is the leftmost agent location in group $j$ and  $x^j_{rm}$ is the rightmost agent location in group $j$.
\end{mechanism}

\begin{theorem} \label{thm:lof}
    LoF-M is strategyproof and has an approximation ratio of $17/8$ when locations are private and group memberships are public.
\end{theorem}
\begin{proof}[proof sketch]
    For the strategyproofness, we can discuss by two cases: agent in group $j^*$, and agent not in group $j^*$. Next, we will show the approximation ratio by case. Let $y^*$ be the optimal facility location, we will start from easier cases to the most difficult ones.

    \paragraph{Case 1: $y=x^*_{lm}$ and $y>y^*$.} \ \\
    \begin{center}
    \begin{tikzpicture}[scale = 2]
        \tikzstyle{every node}=[font=\footnotesize]
            \filldraw[black] (1/2,0) circle (0.5pt);
            \filldraw[black] (3/2,0) circle (0.5pt);
            \node at (0,-1/10) {$0$};
            \node at (1/2,-1/10) {$y^*$};
            \node at (3/2,-1/10) {$y(x^*_{lm})$};
            \node at (2,-1/10) {$1$};           
            \draw [-](0,0)--(2,0);
    \end{tikzpicture}
    \end{center}
    In this case, we can verify that the maximum cost achieved by $y$ is incurred by a group leftmost agent. Suppose that the maximum cost achieved by $y$ is incurred by $x^{j'}_{lm}$. Then we can show that $y$ is within the interval $[x^{j'}_{lm}, x^{j'}_{rm}]$, and the maximum cost achieved by $y^*$ is at least the maximum cost among $x^{j'}_{lm}$ and $x^{j'}_{rm}$. From Lemma~\ref{lem:single group} we have $2$-approximation.

    \paragraph{Case 2: $y=x^*_{lm}$ and $y<y^*$.} \ \\
    \begin{center}
    \begin{tikzpicture}[scale = 2]
        \tikzstyle{every node}=[font=\footnotesize]
            \filldraw[black] (1/2,0) circle (0.5pt);
            \filldraw[black] (3/2,0) circle (0.5pt);
            \node at (0,-1/10) {$0$};
            \node at (3/2,-1/10) {$y^*$};
            \node at (1/2,-1/10) {$y(x^*_{lm})$};
            \node at (2,-1/10) {$1$};           
            \draw [-](0,0)--(2,0);
    \end{tikzpicture}
    \end{center}
    In this case, we can verify that the maximum cost achieved by $y$ is incurred by a group rightmost agent. Then we can use an analysis similar to Case 1 to show $2$-approximation.

    \paragraph{Case 3: $y=x^*_{rm}$ and $y>y^*$.}\ \\
        \begin{center}
    \begin{tikzpicture}[scale = 2]
        \tikzstyle{every node}=[font=\footnotesize]
            \filldraw[black] (1/2,0) circle (0.5pt);
            \filldraw[black] (1,0) circle (0.5pt);
            \node at (0,-1/10) {$0$};
            \node at (1/2,-1/10) {$y^*$};
            \node at (1,-1/10) {$y(x^*_{rm})$};
            \node at (2,-1/10) {$1$};           
            \draw [-](0,0)--(2,0);
    \end{tikzpicture}
    \end{center}
    In this case, we can verfiy that the maximum cost achieved by $y$ is incurred by a group leftmost agent. Then we can use an analysis similar to Case 1 to show $2$-approximation.

    \paragraph{Case 4: $y=x^*_{rm}$ and $y < y^*$.} 
    In this case, we can verify that the maximum cost achieved by $y$ is incurred by a group rightmost agent. Let this group be $G_{jr}$, the leftmost agent in $G_{jr}$ be $x_{lm}^{jr}$, and the rightmost agent in $G_{jr}$ be $x_{rm}^{jr}$ (for simplify the description, we will use $x_i$ to denote both agent $i$ and their location). If $y$ is at or on the right of $x_{lm}^{jr}$, from Lemma \ref{lem:single group} we have the approximation ratio of $2$. Now, we consider the case where $y$ is on the left of $x_{lm}^{jr}$.
    
    Let the group which the agent at $x^*_{rm}$ belongs to be group $G_{jl}$, the leftmost agent in $G_{jl}$ be $x_{lm}^{jl}$, and the rightmost agent in $G_{jl}$ be $x_{rm}^{jl}$. We can show that moving all agents in $G_{jl}$ to $y$ iteratively will not decrease the approximation ratio. After the movement, we will examine which agent incurs the maximum cost achieved by $y^*$. There are two possible cases, and we will analyze that the ratio in each of these cases does not exceed $17/8$.

    \textbf{Case 4-1.} $y^*$ is on the left of $x_{lm}^{jr}$. \\
    \begin{center}
    \begin{tikzpicture}[scale = 2]
        \tikzstyle{every node}=[font=\footnotesize]

            \filldraw ([xshift=-2pt,yshift=-0.75pt]0.6,0) rectangle ++(1.5pt,1.5pt);
            \filldraw[black] (2,0) circle (0.75pt);
            \filldraw[black] (2.5,0) circle (0.75pt);
            \node at (0,-0.15) {$0$};
            \node at (0.6,-0.15) {$(y, x^{jl}_{lm}, x^{jl}_{rm})$};
            \node at (1.5,-0.15) {$y^*$};
            \node at (2,-0.15) {$x^{jr}_{lm}$};
            \node at (2.5,-0.15) {$x^{jr}_{rm}$};
            \node at (3.5,-0.15) {$1$};           
            \draw [-](0,0)--(3.5,0);
    \end{tikzpicture}
    \end{center}
    In this case, we focus on the impact of agent $x_{lm}^{jl}$ and agent $x_{rm}^{jr}$ on the maximum cost achieved by $y^*$. We will increase the approximation ratio by minimizing the maximum cost achieved by $y$. Since all agents in $G_{jl}$ are at the same location, the cost of agent $x_{lm}^{jl}$ is
        \begin{align*}
             c_{lm}^{jl}(y^*, \mathbf{r}) \!=\! d(y^*,x_{rm}^{jr})\!+\!\alpha_{jl}(|G_{jl}|\!-\!1)(1\!-\!d(y^*,x_{rm}^{jr}))\\
             = (1-\alpha_{jl}(|G_{jl}|-1))d(y^*,x_{rm}^{jr})+\alpha_{jl}(|G_{jl}|-1).
        \end{align*}

        Let $\alpha_{jl}(|G_{jl}|-1)$ be $c_{jl}\in [0,1]$. We further have
        \begin{align*}
            (1-c_{jl})d(y^*,x_{rm}^{jr})+c_{jl} \!-\! ((1-c'_{jl})d(y^*,x_{rm}^{jr})+c'_{jl}) \\
             = (1-d(y^*,x_{rm}^{jr}))(c_{jl}-c'_{jl}) \ge 0,
        \end{align*}
        where $c_{jl}>c'_{jl}$, implying that decreasing the value of $c_{jl}$ will not increase the maximum cost achieved by $y^*$ and not change the maximum cost achieved by $y$. Hence, the approximation ratio can only be larger by setting $c_{jl}=0$. Then we have $c_{lm}^{jl}(y^*, \mathbf{r}) = d(y^*,x_{rm}^{jr})$.
        Moreover, the cost of agent $x_{rm}^{jr}$ will increase the amount of at most $(y^*-y)-\alpha_{jr}(|G_{jr}|-1)(y^*-y)$ from $y^*$ to $y$ since all agents in $G_{jr}$ are on the right of $y^*$.Therefore, the approximation ratio is at most
        \begin{align*}
            \rho &\le \frac{c_{lm}^{jl}(y^*,\mathbf{r}) + (y^*-y)-\alpha_{jr}(|G_{jr}|-1)(y^*-y)}{c_{lm}^{jl}(y^*,\mathbf{r})}\\
            & \le \frac{2-\alpha_{jr}(|G_{jr}|-1)}{1} \le 2.
        \end{align*}
        
        \textbf{Case 4-2.} $y^*$ is on the right of $x_{lm}^{jr}$. \\
        \begin{center}
            \begin{tikzpicture}[scale = 2]
        \tikzstyle{every node}=[font=\footnotesize]

            \filldraw ([xshift=-2pt,yshift=-0.75pt]0.6,0) rectangle ++(1.5pt,1.5pt);
            \filldraw[black] (2,0) circle (0.75pt);
            \filldraw[black] (3,0) circle (0.75pt);
            \node at (0,-0.15) {$0$};
            \node at (0.6,-0.15) {$(y, x^{jl}_{lm}, x^{jl}_{rm})$};
            \node at (2,-0.15) {$x^{jr}_{lm}$};
            \node at (3,-0.15) {$x^{jr}_{rm}$};
            \node at (2.5,-0.15) {$y^*$};
            \node at (3.5,-0.15) {$1$};           
            \draw [-](0,0)--(3.5,0);

            \draw [decorate,decoration={brace,amplitude=5pt,mirror,raise=4ex}](0.6,0) -- (2,0) node[midway,yshift=-3em]{$d_1$};
            \draw [decorate,decoration={brace,amplitude=5pt,mirror,raise=4ex}](2,0) -- (2.5,0) node[midway,yshift=-3em]{$d_2$};
    \end{tikzpicture}
    \end{center}
        We consider the impact of agents $x_{lm}^{jl}$, $x_{lm}^{jr}$, and $x_{rm}^{jr}$ on the maximum cost achieved by $y^*$. First, we can use a similar analysis to show that setting $c_{jl}=0$ makes the approximation ratio larger. From Lemma~\ref{lem:single group} we know that the agents $x_{lm}^{jr}$ and $x_{rm}^{jr}$ will have the same cost with respect to the location $(x_{lm}^{jr} + x_{rm}^{jr})/2$. 
        
        If the cost of agent $x_{lm}^{jl}$ is larger than the other two agents with respect to the location  $(x_{lm}^{jr} + x_{rm}^{jr})/2$.. We can move $x_{lm}^{jr}$ to the left until it reaches $y$, or the cost of agent $x_{lm}^{jl}$ is equal to the other two agents with respect to the location  $(x_{lm}^{jr} + x_{rm}^{jr})/2$.. After the movement, the approximation ratio will not decrease. If the movement is terminated by reaching $y$, we can use Lemma~\ref{lem:single group} to show an approximation ratio $2$. 
        
        Next, we consider the case where the movement is terminated by the other condition. In this case, all three agents $x_{lm}^{jl}$, $x_{lm}^{jr}$, $x_{rm}^{jr}$ have the same cost with respect to the location $(x_{lm}^{jr} + x_{rm}^{jr})/2$. Then we focus on the approximation ratio between the maximum cost achieved by $y$ and $(x_{lm}^{jr} + x_{rm}^{jr})/2$ since the approximation ratio between the maximum cost achieved by $y$ and $y^*$ is less or equal to the former one. To simplify the description, let $d_1=x_{lm}^{jr}-x_{lm}^{jl}$ and $d_2 = (x_{lm}^{jr} + x_{rm}^{jr})/2$, and we just use $y^*$ for $(x_{lm}^{jr} + x_{rm}^{jr})/2$. First we have the maximum cost achieved by $y^*$ is at least
        \begin{equation*} 
            \mc(y^*,\mathbf{r}) \ge c_{lm}^{jl}(y^*,\mathbf{r}) = d_1+d_2.
        \end{equation*}
        Let $\alpha_{jr}(|G_{jr}|-1)$ be $c_{jr}$, we also have the maximum cost achieved by $y^*$ is at least
        \begin{equation*}
            \mc(y^*,\mathbf{r}) \ge c_{rm}^{jr}(y^*,\mathbf{r}) \ge d_2 + c_{jr}(1-d_2).
        \end{equation*}

        From $y^*$ to $x_{lm}^{jr}$, the cost of agent $x_{rm}^{jr}$ will increase by the amount of at most $(d_2+c_{jr}d_2)$, since there are at most $|G_{jr}|-1$ agents at $x_{lm}^{jr}$.  From $x_{lm}^{jr}$ to $y$, the cost of agent $x_{rm}^{jr}$ will increase by the amount of at most $(d_1-c_{jr}d_1)$, since all agents in $G_{jr}$ are at or on the right of $x_{lm}^{jr}$. Hence, the difference between the maximum cost achieved by $y$ and $y^*$ is at most
        \begin{equation*}
            \mc(y,\mathbf{r}) \le \mc(y^*,\mathbf{r}) + d_1 + d_2 + c_{jr}(d_2-d_1).
        \end{equation*}

        Then, we have the approximation ratio
        \begin{equation} \label{eq: rho*1}
            \rho = \frac{\mc(y,\mathbf{r})}{\mc(y^*,\mathbf{r})} \le 1+\frac{d_1 + d_2 + c_{jr}(d_2-d_1)}{d_1+d_2},
        \end{equation}
        and 
        \begin{equation} \label{eq: rho*2}
            \rho = \frac{\mc(y,\mathbf{r})}{\mc(y^*,\mathbf{r})} \le 1+\frac{d_1 + d_2 + c_{jr}(d_2-d_1)}{d_2 + c_{jr}(1-d_2)}.
        \end{equation}

        One can verify that Equation \ref{eq: rho*1} is monotonically decreasing with respect to $d_1$ and Equation \ref{eq: rho*2} is monotonically increasing with respect to $d_1$. The equality holds when $d_1=c_{jr}(1-d_2)$. By plugging $d_1$ into both equations, we have the approximation ratio
        \begin{align*}
            \rho & = \le 1+\frac{c_{jr}(1-d_2) + d_2 + c_{jr}(d_2-c_{jr}(1-d_2))}{d_2 + c_{jr}(1-d_2)}\\
            & =  1 + \frac{(1+c^2_{jr})d_2+c_{jr}-c^2_{jr}}{(1-c_{jr})d_2+c_{jr}},
        \end{align*}
        which is monotonically increasing with respect to $d_2$. Note that $d_1+2d_2\le 1$, we have $d_2\le \frac{1-c_{jr}}{2-c_{jr}}$. By plugging $d_2$ into the equation, we have the approximation ratio which is only related to $c_{jr}$. By carefully calculation, the approximation ratio reaches the maximum of $17/8$ when $c_{jr}=1/4$.
\end{proof}
Since \citet{procaccia2009approximate} is a special case of our setting where all $\alpha_j=0$, their lower bound of $2$ still hold in our setting.

\subsection{Public Locations and Private Group Memberships}
In this subsection, we locate the facility at the middle point between the location of the leftmost agent and the location of the rightmost agent.

\begin{mechanism}[Middle-point Mechanism (Mid-M)]
    Given any profile set $\mathbf{r}$, locate the facility at $y = \frac{1}{2}(\min_{i \in N}\{x_i\} + \max_{i \in N}\{x_i\})$.
\end{mechanism}

\begin{theorem}~\label{thm:mpm}
    Mid-M is strategyproof, and has an approximation ratio of $(29+20\sqrt{10})/54$  for minimizing the maximum cost when locations are public and group memberships are private.
\end{theorem}
\begin{proof}[proof sketch]
    Since Mid-M does not use the group memberships, it satisfies the strategyproofness trivially. Hence, we focus on the approximation ratio. Let the optimal facility location be $y^*$. Without loss of generality, we assume that $y< y^*$, the leftmost agent is in group $G_{l}$, and the maximum cost achieved by $y$ is incurred by an agent in $G_r$. To simplify the description, we use $x_{lm}^j$ and $x_{rm}^j$ to denote the leftmost and the rightmost agent in group $G_j$. First, we can see that the maximum cost achieved by $y$ is incurred by agent $x_{rm}^r$. 
    Next, we focus on the impact of agent $x_{lm}^{l}$, $x_{rm}^{r}$ on the maximum cost achieved by $y^*$.  First, we can move all agents in $G_l$ to $x_{lm}^{l}$, where the approximation ratio will not decrease. 
    
    If the cost of agent $x_{lm}^{l}$ is larger than the other two agents with respect to the location $(x_{lm}^{r} + x_{rm}^{r})/2$. We can use a similar analysis with Theorem~\ref{thm:lof} to move $x_{lm}^{jr}$ to the left until it reaches $x_{lm}^{l}$, or the cost of agent $x_{lm}^{l}$ is equal to the other two agents with respect to the location $(x_{lm}^{r} + x_{rm}^{r})/2$. After the movement, the approximation ratio will not decrease. If the movement is terminated by reaching $y$, we have $y \ge (x_{lm}^{r} + x_{rm}^{r})/2$ while $y^* \le (x_{lm}^{r} + x_{rm}^{r})/2$. Hence, we have $y=y^*$ and the approximation ratio is $1$.

    Next, we consider the case where the movement is terminated by the other condition. In this case, all three agents $x_{lm}^{l}$, $x_{lm}^{r}$, $x_{rm}^{r}$ have the same cost with respect to the location $(x_{lm}^{r} + x_{rm}^{r})/2$. Then we focus on the approximation ratio between the maximum cost achieved by $y$ and $(x_{lm}^{r} + x_{rm}^{r})/2$ since the approximation ratio between the maximum cost achieved by $y$ and $y^*$ is less or equal to the former one. To simplify the description, let $d_1=y-x_{lm}^{l}$, $d_2 = (x_{lm}^{r} + x_{rm}^{r})/2$, and $d_3 = d(y, x_{lm}^{r})$. We just use $y^*$ for $(x_{lm}^{r} + x_{rm}^{r})/2$. 
    \paragraph{Subcase 1.} $y\le x_{lm}^{r}$
    First we have the maximum cost achieved by $y^*$ is at least
        \begin{equation*} 
            \mc(y^*,\mathbf{r}) \ge c_{lm}^{l}(y^*,\mathbf{r}) = d_1+d_2+d_3.
        \end{equation*}
        Let $\alpha_{r}(|G_{r}|-1)$ be $c_{r}$, we also have the maximum cost achieved by $y^*$ is at least
        \begin{equation*}
            \mc(y^*,\mathbf{r}) \ge c_{rm}^{r}(y^*,\mathbf{r}) \ge d_2 + c_{r}(1-d_2).
        \end{equation*}

        From $y^*$ to $x_{lm}^{r}$, the cost of agent $x_{rm}^{r}$ will increase by the amount of at most $(d_2+c_{r}d_2)$, since there are at most $|G_{r}|-1$ agents at $x_{lm}^{r}$.  From $x_{lm}^{r}$ to $y$, the cost of agent $x_{rm}^{jr}$ will increase by the amount of at most $(d_3-c_{r}d_3)$, since all agents in $G_{r}$ are at or on the right of $x_{lm}^{r}$. Hence, the difference between the maximum cost achieved by $y$ and $y^*$ is at most
        \begin{equation*}
            \mc(y,\mathbf{r}) \le \mc(y^*,\mathbf{r}) + d_3 + d_2  + c_{r}(d_2-d_3).
        \end{equation*}

        Then, we have the approximation ratio
        \begin{equation*}
            \rho = \frac{\mc(y,\mathbf{r})}{\mc(y^*,\mathbf{r})} \le 1+\frac{d_3 + d_2  + c_{r}(d_2-d_3)}{d_1+d_2+d_3}.
        \end{equation*}
        Since $y$ is the middle point, we have $d_1\ge d_3+2d_2$. Then we have 
        \begin{equation} \label{eq: rho*3}
            \rho = \frac{\mc(y,\mathbf{r})}{\mc(y^*,\mathbf{r})} \le 1+\frac{d_3 + d_2  + c_{r}(d_2-d_3)}{3d_2+2d_3},
        \end{equation}
        Moreover, we have
        \begin{equation} \label{eq: rho*4}
            \rho = \frac{\mc(y,\mathbf{r})}{\mc(y^*,\mathbf{r})} \le 1+\frac{d_3 + d_2  + c_{r}(d_2-d_3)}{d_2 + c_{r}(1-d_2)}.
        \end{equation}
        One can verify that Equation \ref{eq: rho*3} is monotonically decreasing with respect to $d_3$ and Equation \ref{eq: rho*4} is monotonically increasing with respect to $d_3$. The equality holds when $d_3=(c_r-c_rd_2-2d_2)/2$. By plugging $d_3$ into both equations, we have the approximation ratio
        \begin{align*}
            \rho & =1+\frac{c_{r}^2d_2 + c_{r}  + 3c_{r}d_2-c_{r}^2}{2d_2 + 2c_{r}(1-d_2)} \le 1+\frac{c_{r}^3 - 5c_{r}^2  + 5c_{r}}{2},
        \end{align*}
        which is monotonically increasing with respect to $d_2$. Note that $d_1+d_3+2d_2\le 1$ and $d_1\ge d_3+2d_2$, we have $d_2\le (1-c_r)/(2-c_r)$. By plugging $d_2$ into the equation, we have the approximation ratio which is only related to $c_{r}$. By carefully calculation, the approximation ratio reaches the maximum of $\frac{29+20\sqrt{10}}{54}$ when $c_{r}=\frac{5-\sqrt{10}}{3}$.

        \paragraph{Subcase 2.} $y> x_{lm}^{r}$. We can use an anslysis similar to Subcase 1 to show the approximation ratio of at most $14/9$.
\end{proof}
Next, we will show the approximation ratio in this setting.

\begin{theorem}~\label{thm:8}
    Any deterministic mechanism has an approximation ratio of at least $2(1+2\sqrt{2})/7$ for minimizing the maximum cost when locations are public and group memberships are private.
\end{theorem}

\subsection{Private Locations and Group Memberships}
When both locations and group memberships are private, Lof-M and Mid-M are no longer strategyproof. Hence, we just use Left-M in this setting. Moreover, the lower bound is at least $2$, since misreporting only the location is a special case of misreporting both.

\section{Conclusion and Future Work}~\label{sec:conclusion}
We investigate mechanism design for congested facility location, where agents are partitioned into groups based on conflicting interests, and their costs increase due to congestion when the facility is closer to their competitors. We address three misreporting settings---location only, group membership only, and both---and propose strategyproof mechanisms that achieve near-tight approximation bounds for minimizing the social cost (\(\soc\)) and maximum cost (\(\mc\)).

Our analysis opens several avenues for future research. First, incorporating inter-group congestion could model scenarios where all agents compete for limited resources, enhancing the model’s applicability. Second, exploring group fairness objectives, such as minimizing the maximum group cost, could promote equitable resource allocation. Finally, allowing agents to belong to multiple groups would better capture complex affiliations in social or distributed systems.
\clearpage

\printbibliography

\clearpage
\appendix
\section{Missing Proofs}
\subsection*{Proof of \Cref{pro:aversion}}
The cost function \( c_i(y, \mathbf{r}) = d(y, x_i) + \alpha_{g_i} \sum_{k \in G_{g_i} \setminus \{i\}} (1 - d(y, x_k)) \) satisfies the \emph{Consistency of Aversion} property if and only if \( \alpha_j \leq 1/(|G_j|-1) \) for each group \( j \in [m] \).
\begin{proof}
Fix a group $G_j$, and consider an agent $i \in G_j$ located at $x_i$. Define the cost function as:
\begin{align*}
  c_i(y) &= d(y, x_i) + \alpha_j \sum_{k \in G_j \setminus \{i\}} (1 - d(y, x_k)) \\&= |y - x_i| + \alpha_j \sum_{k \in G_j \setminus \{i\}} (1 - |y - x_k|).
\end{align*}

To satisfy the \emph{Consistency of Aversion} property, the cost $c_i(y)$ must be non-decreasing as $y$ moves from $x_i$ to the competing group members. Let us compute the derivative of $c_i(y)$ with respect to $y$. For $y \notin \{x_i\} \cup \{x_k : k \in G_j \setminus \{i\} \}$, $c_i(y)$ is differentiable, and its derivative is:
\[
\frac{d}{dy} c_i(y) = \operatorname{sign}(y - x_i) - \alpha_j \sum_{k \in G_j \setminus \{i\}} \operatorname{sign}(y - x_k),
\]
where $\operatorname{sign}(z) = 1$ if $z > 0$, $-1$ if $z < 0$.

We use
$L(y) = |\{k \in G_j \setminus \{i\} : x_k < y\}|$ to define the number of competitors to the left of $y$,
$R(y) = |\{k \in G_j \setminus \{i\} : x_k > y\}|$ to define the number of competitors to the right of $y$,
so that
\[
\sum_{k \in G_j \setminus \{i\}} \operatorname{sign}(y - x_k) = L(y) - R(y).
\]

Now consider two case:

\textbf{Case 1: $\operatorname{sign}(y - x_i) = 1$.}  We want:
\begin{equation*}
    \frac{d}{dy} c_i(y) = 1 - \alpha_j (L(y) - R(y))\geq1-\alpha(|G_j|-1) \geq 0.
\end{equation*}

To ensure non-negativity, we must have:
\[
\alpha_j \leq \frac{1}{|G_j| - 1}.
\]

\textbf{Case 2:$\operatorname{sign}(y - x_i) = -1$.} Then:
\[
\frac{d}{dy} c_i(y) = -1 - \alpha_j (L(y) - R(y))\leq-1 +\alpha_j (|G_j|-1)\leq 0.
\]

To ensure non-negativity, we must have:
\[
\alpha_j \leq \frac{1}{|G_j| - 1}.
\]

Therefore, the Consistency of Aversion condition only requires cost to be non-decreasing when the facility moves toward the competing agents. This leads to the necessary and sufficient condition:
\[
\alpha_j \leq \frac{1}{|G_j| - 1}.
\]
\end{proof}

\subsection*{Proof of Theorem~\ref{thm:med-social}}
    Median-point Mechanism is strategyproof when both locations and group memberships are private, and has an approximation ratio of $2$.

\begin{proof}
  Through the definition, we know that the ideal facility location for each agent is their own location. Med-M outputs the median agent locations of all agents. First, misreporting group memberships cannot change the facility location. Then, we consider misreporting locations. For the median agent, Med-M locates the facility at their favourite location, thus they has no incentive to misreport. For the agents on the left side of the facility location, they cannot move the facility location until they misreport their location to the right of the facility location, which makes the facility move farther away from them. Hence, they has no incentive to misreport. For the agent on the right side of the facility location, we can use a similar way to analysis them. Hence, Med-M is strategyproof. Next, we consider its approximation ratio.

  Without loss of generality we assume that the optimal facility location $y^*>y$ ($y^*=y$ implies the approximation ratio 1). Note that the social cost achieved by $y$ can be reformulated as
  \begin{align*}
      & \soc(y, \mathbf{r}) = \sum_{i\in N}\left (d(x_i,y)+\alpha_{g_i}\sum_{k\in G_{g_i};k\neq i}(1-d(y,x_k))\right )\\
      &= \sum_{i\in N}\Big( (1-\alpha_{g_i}(|G_{g_i}|-1))d(y,x_i)+\alpha_{g_i}(|G_{g_i}|-1)\Big ).
  \end{align*}
  Let the contribution of each agent $i$ to the social cost is $(1-\alpha_{g_i}(|G_{g_i}|-1))d(y,x_i)+\alpha_{g_i}(|G_{g_i}|-1)$, we have
  \begin{align*}
     &\soc(y,\mathbf{r})-\soc(y^*,\mathbf{r}) \\
      = & \sum_{i\in N} (w_i d(y,x_i) + C_i) - \sum_{i\in N} (w_i d(y^*,x_i) + C_i)\\
      = & \sum_{i\in N} w_i (d(y, x_i)-d(y^*, x_i)),
  \end{align*}
 where $w_i = 1-\alpha_{g_i}(|G_{g_i}|-1)$ and $C_i = \alpha_{g_i}(|G_{g_i}|-1)$.
 For agents at or on the right of $y$, we have $d(y, x_i)-d(y^*, x_i) \le y^*-y$, where the equality holds when agent $i$ is at or on the right of $y^*$. For agents on the left of $y$, we have $d(y, x_i)-d(y^*, x_i) = y-y^*$. Hence, we futher have
 \begin{align*}
     \soc(y,\mathbf{r})-\soc(y^*,\mathbf{r}) 
     & \le \sum_{x_i> y} w_i (y^*-y) + \sum_{x_i\le y} w_i (y-y^*)\\
     & = (\sum_{x_i> y} w_i - \sum_{x_i\le y} w_i)(y^*-y),
 \end{align*}

 Moreover, the social cost incurred by $y^*$ is at least
 \begin{align*}
     \soc(y^*, \mathbf{r}) &\ge\!\! \sum_{x_i\le y} \!(w_i d(x_i, y^*) \!+\! C_i) + \sum_{x_i> y} \!(w_i d(x_i, y^*) \!+\! C_i)\\
     &\ge\! \sum_{x_i\le y} \!(w_i (y^*-y) \!+\! C_i)+\sum_{x_i> y} \!C_i.
 \end{align*}
 
 Then we have the approximation ratio is at most
 \begin{align*}
    \rho \le & \frac{\soc(y^*,\mathbf{r})+(\sum_{x_i> y} w_i - \sum_{x_i\le y} w_i)(y^*-y)}{\soc(y^*,\mathbf{r})}\\
     = & \frac{\sum_{x_i> y} w_i(y^*-y)+\sum_{x_i\le y} C_i + \sum_{x_i> y} C_i}{\sum_{x_i\le y} w_i (y^*-y) +  \sum_{x_i\le y} C_i + \sum_{x_i> y} C_i}\\
     \le & \frac{\sum_{x_i> y} w_i+\sum_{x_i\le y} C_i + \sum_{x_i> y} C_i}{\sum_{x_i\le y} w_i +  \sum_{x_i\le y} C_i + \sum_{x_i> y} C_i}\\
     = & \frac{\sum_{x_i> y} 1+\sum_{x_i\le y} C_i}{\sum_{x_i\le y} 1 + \sum_{x_i> y} C_i}.
 \end{align*}

 Let $c_{\max} = \max_{j \in [m]}\alpha_j(|G_j|-1)$ and $c_{\min} = \min_{j \in [m]}\alpha_j(|G_j|-1)$. Since $y$ is the median location, the approximation ratio is at most
 \begin{align*}
     \rho &= \frac{\sum_{x_i> y} 1+\sum_{x_i\le y} C_i}{\sum_{x_i\le y} 1 + \sum_{x_i> y} C_i} \le \frac{\frac{n}{2}+\frac{n}{2}c_{\max}}{\frac{n}{2}+\frac{n}{2}c_{\min}}=\frac{1+c_{\max}}{1+c_{\min}}.
 \end{align*}

 Since $\alpha_j (|G_j|-1) \in [0,1]$, we have the approximation ratio of $2$.
\end{proof}



\subsection*{Proof of Theorem~\ref{thm:res}}

    Restructure Mechanism is strategyproof and optimal for minimizing the social cost.

\begin{proof}
    For each agent, the only way to change the facility location is to misreport their location to the opposite side of \( x_i^* \), which eventually leads the facility to be farther away from their true location. Hence, their cost increases after misreporting. Therefore, they has no incentive to misreport their location. 

    From Proposition~\ref{pro:opt} we have $x_i^*$ is the optimal facility location.

\end{proof}
\subsection*{Proof of Lemma~\ref{lem:partial group SP}}
  If a deterministic mechanism is strategyproof, then for any set of agents at the same location, each individual cannot benefit if they misreport their group memberships simultaneously.
\begin{proof}
    We denote the set of agents at the same location $x$ as $S\subseteq N$ with true profiles $r_1,...r_s$. Let $\mathbf{r}_{-S}$ denote a collection of true profiles of $n$ agents except for agents in $S$. Consider a series of profiles $\mathbf{r}^i (0\leq i\leq s)$ where $s-i$ agents with $r_{i+1},...,r_{s}$ in $S$ misreport their group memberships simultaneously. We have
  \begin{equation*}
      \begin{split}
            \mathbf{r}^i&=\{r_1,...,r_i,r'_{i+1},...,r'_s\}\cup \mathbf{r}_{-S},\\
            \mathbf{r}^{i\!-1}\!\!&=\{r_1,...,r_{i-1},r'_i,...,r'_s\}\cup \mathbf{r}_{-S}. 
      \end{split}
  \end{equation*}
  Therefore $\mathbf{r}^{i-1}$ could be regarded as the agent $i$ in $\mathbf{r}^i$ misreporting the profile to $r'_i$. From the definition of strategyproofness, we have $c_i(f(\mathbf{r}^i), \mathbf{r}^i) \le c_i(f(\mathbf{r}^{i-1}), \mathbf{r}^i)$, implying that $d(f(\mathbf{r}^i), x) \le d(f(\mathbf{r}^{i-1}), x)$. 
  By integrating all the inequalities achieved by $i\in[0,s]$, we further have 
  \begin{equation*}
  \begin{aligned}
      d(f(\mathbf{r}^s), x) \le d(f(\mathbf{r}^{s-1}), x) \le \cdots \le d(f(\mathbf{r}^{0}), x),
    \end{aligned}
  \end{equation*}
  which completes the proof.
\end{proof}
\subsection*{Proof of Theorem~\ref{thm:lbp}}
    Any deterministic mechanism which outputs the facility at an agent location has an approximation ratio at least $2$ for minimizing the social cost when locations are public and group memberships are private.

\begin{proof}
    Given any strategyproof mechanism, consider a profile set $\mathbf{r}$ where all the agents' locations $x_1=\cdots=x_{\frac{n}{2}}=0$ and $x_{\frac{n}{2}+1}=\cdots=x_n=1$ ($n$ is an even number). Suppose that all agents belong to $G_1$ where $\alpha_1=0$. Due to the symmetry, we assume that the mechanism will output $y=1$.
    Now we consider another profile set $\mathbf{r}'$ where for all $\frac{n}{2}$ agents at $1$, $\alpha_2=1/(n/2-1)$. In this profile set, the optimal location is $0$. To achieve the approximation ratio less than $2$, the mechanism should output $y = 0$. Hence, agents at $1$ in $\mathbf{r}'$ can benefit by misreporting their groups together, which makes the facility from $0$ to $1$. From Lemma~\ref{lem:partial group SP} we have that mechanism is not strategyproof. 
\end{proof}
\subsection*{Proof of Lemma~\ref{lem:partial group SP location}}
 If a deterministic mechanism is strategyproof, then for any set of agents at the same location, each individual cannot benefit if they misreport their locations simultaneously.
\begin{proof}
    We denote the set of agents at the same location $x$ as $S\subseteq N$ with true profiles $r_1,...r_s$. Let $\mathbf{r}_{-S}$ denote a collection of true profiles of $n$ agents except for agents in $S$. Consider a series of profiles $\mathbf{r}^i (0\leq i\leq s)$ where $s-i$ agents with $r_{i+1},...,r_{s}$ in $S$ misreport their locations simultaneously. We have
  \begin{equation*}
      \begin{split}
            \mathbf{r}^i&=\{r_1,...,r_i,r'_{i+1},...,r'_s\}\cup \mathbf{r}_{-S},\\
            \mathbf{r}^{i\!-1}\!\!&=\{r_1,...,r_{i-1},r'_i,...,r'_s\}\cup \mathbf{r}_{-S}. 
      \end{split}
  \end{equation*}
  Therefore $\mathbf{r}^{i-1}$ could be regarded as the agent $i$ in $\mathbf{r}^i$ misreporting the profile to $r'_i$. From the definition of strategyproofness, we have $c_i(f(\mathbf{r}^i), \mathbf{r}^i) \le c_i(f(\mathbf{r}^{i-1}), \mathbf{r}^i)$, implying that $d(f(\mathbf{r}^i), x) \le d(f(\mathbf{r}^{i-1}), x)$. 
  By integrating all the inequalities achieved by $i\in[0,s]$, we further have 
  \begin{equation*}
  \begin{aligned}
      d(f(\mathbf{r}^s), x) \le d(f(\mathbf{r}^{s-1}), x) \le \cdots \le d(f(\mathbf{r}^{0}), x),
    \end{aligned}
  \end{equation*}
  which completes the proof.
\end{proof}
\subsection*{Proof of Theorem~\ref{thm:pp2}}

    Any deterministic mechanism has an approximation ratio at least $2$ for minimizing the social cost when both locations and group memberships are private.

\begin{proof}
        Given any strategyproof mechanism, consider a profile set where all the agents' locations $x_1=\cdots=x_{\frac{n}{2}}=1/2-\epsilon$ and $x_{\frac{n}{2}+1}=\cdots=x_n=1/2+\epsilon$ ($n$ is an even number). Suppose that all agents belong to $G_1$ where $\alpha_1=0$. Due to the symmetry, we assume that the mechanism will output $y\in [1/2,1]$. We first show that the mechanism has to output $y\in [1/2, 1/2+\epsilon] = [x_2-\epsilon, x_2]$. To see this, if the mechanism output $y\in (1/2+\epsilon, 1]$. Agents at $0$ can benefit by misreporting their locations to $(1/2+\epsilon)$ together, which makes the facility to be located at $(1/2+\epsilon)$, which is closer to them.
        Now we consider another profile set where all $\frac{n}{2}$ agents at $x_1=1/2-\epsilon$ move to $x_1'=1/2-3\epsilon$. To achieve the strategyproofness, the mechanism has to output $y\in [1/2-3\epsilon, 1/2-2\epsilon] = [x_1',x_1'+\epsilon]$. By using a series of similar analysis, there is a profile set where $x_1=0$ or $x_2=1$, and $x_2-x_1 > \frac{1}{2}$, the mechanism outputs $y\in [x_1,x_1+\epsilon]\cup [x_2-\epsilon, x_2]$. Without loss of generality we assume that for the profile set $\mathbf{r}$, all the agents' locations $x_1=\cdots=x_{\frac{n}{2}}=0$ and $x_{\frac{n}{2}+1}=\cdots=x_n \ge 1/2+\epsilon$ ($n$ is an even number). Suppose that all agents belong to $G_1$ where $\alpha_1=0$. If the output of the mechanism is $y\in [x_2-\epsilon, x_2]$, consider another profile set $\mathbf{r}'$ where all agents at $x_2$ belong to group $G_2$ where $\alpha_2=1/(|G_2|-1)$. To achieve the approximation ratio of less than $2$, the mechanism has to output $y < x_2-\epsilon$. In this case, agents at $x_2$ in $\mathbf{r}'$ can benefit by misreporting their groups together, which makes the facility closer. From Lemma~\ref{lem:partial group SP} we have that mechanism is not strategyproof. If the output of the mechanism is $y\in [0, \epsilon]$, we can use a similar analysis to show the approximation ratio.
\end{proof}
\subsection*{Proof of Lemma~\ref{lem:single group}}

    Locating the facility at any point between the leftmost agent location and the rightmost agent location can achieve an approximation ratio of $2$, when there is only one group.

\begin{proof}
   Let $lm$ and $rm$ be the leftmost agent and the rightmost agent, respectively. Let the facility location be $y$. We will show that $c_{lm}(y,\mathbf{r}) > c_{i}(y,\mathbf{r})$ for any $x_{lm} < x_{i} \le y$ and $c_{rm}(y,\mathbf{r}) > c_{i}(y,\mathbf{r})$ for any $y < x_{i} \le x_{rm}$. Hence, no matter at which point between $x_{lm}$ and $x_{rm}$ the facility is located, the maximum cost is incurred by either the cost of the leftmost agent or the rightmost agent. To see this, we first show that $c_{lm}(y,\mathbf{r}) > c_{i}(y,\mathbf{r})$ for any $x_{lm} < x_{i} \le y$. To simplify the description, we directly use $\alpha$ and $G$ instead of $\alpha_j$ and $G_j$ since there is only one group.

    We will show that $c_i(x_{rm},\mathbf{r}) > c_{i'}(x_{rm},\mathbf{r})$ for any $x_i < x_{i'} \le x_{rm}$.
    
    The cost of $lm$ is
    \begin{equation*}
        c_{lm}(y,\mathbf{r})=y-x_{lm}+\alpha(1+x_{i}-y)+\alpha\sum_{k\ne lm,i}(1-d(y,x_k)).
    \end{equation*}
    The cost of $x_i$ is
    \begin{equation*}
        c_{i}(y,\mathbf{r})=y-x_{i}+\alpha(1+x_{lm}-y)+\alpha\sum_{k\ne lm,i}(1-d(y,x_k)).
    \end{equation*}
    Hence, we have
    \begin{equation*}
      c_{lm}(y,\mathbf{r})-  c_{i}(y,\mathbf{r})=(\alpha+1)(x_{i}-x_{lm})>0.
    \end{equation*}
    We can also use the similar analysis to show that $c_{rm}(y,\mathbf{r}) > c_{i}(y,\mathbf{r})$ for any $y < x_{i} \le x_{rm}$. Moreoer, we can conclude that the optimal facility location for minimizing the maximum cost is the location where the leftmost agent and the rightmost agent have the same cost. 
    Let  $c_{lm}(y^*,\mathbf{r})=c_{rm}(y^*,\mathbf{r})$, we have
    \begin{align*}
     & y^*\!-\!x_{lm}\!+\!\alpha(1-x_{rm}+y^*)\!+\!\alpha\sum_{k\ne lm,rm}(1-d(y^*,x_k))\\
    = &x_{rm}\!-\!y^*\!+\!\alpha(1+x_{lm}-y^*)\!+\!\alpha\sum_{k\ne lm,rm}(1-d(y^*,x_k)).
    \end{align*}

    Then, we have
    \begin{equation*}
    2y^*(1+\alpha_{g_i})=(1+\alpha_{g_i})(x_{lm}+x_{rm}) \text{ and } y^*=\frac{x_{lm}+x_{rm}}{2},
    \end{equation*}
    and the maximum cost achieved by $y^*$ is 
    \begin{equation*}
        c_{lm}(y^*,\mathbf{r})= \frac{x_{rm}-x_{lm}}{2} + \alpha(|G|-1)(1-\frac{x_{rm}-x_{lm}}{2}).
    \end{equation*}
    From the definition of the agent cost function, we have $c_{rm}(x_{lm}, \mathbf{r}) \ge c_{rm}(y, \mathbf{r})$ and $c_{lm}(x_{rm}, \mathbf{r}) \ge c_{lm}(y, \mathbf{r})$ when $x_{lm}\le y \le x_{rm}$. We only need to calculate
    \begin{equation*}
        \frac{\max\{c_{lm}(x_{rm}, \mathbf{r}), c_{rm}(x_{lm}, \mathbf{r})\}}{\max\{c_{lm}(y^*, \mathbf{r}), c_{rm}(y^*, \mathbf{r})\}}.
    \end{equation*}
    Due to the symmetry, let $c_{lm}(x_{rm}, \mathbf{r})$ be the larger one. Now we focus on the cost of agent $lm$. From $y^*$ to $x_{rm}$, the distance of agent $lm$ is increased by the amount of $x_{rm}-y^*$, the distance of agent $rm$ is decreased by the amount of $x_{rm}-y^*$, and there are at most $|G|-2$ other agents whose distances are decreased by the amount of at most $x_{rm}-y^*$. Hence, we have
    \begin{equation*}
        c_{lm}(x_{rm}, \mathbf{r})-c_{lm}(y^*, \mathbf{r}) \le \alpha(|G|-2)(x_{rm}-y^*).
    \end{equation*}

    Then approximation ratio is
    \begin{align*}
        \rho & = \frac{c_{lm}(x_{rm}, \mathbf{r})}{c_{lm}(y^*, \mathbf{r})}
        \le 1+\frac{\alpha(|G|-2)(x_{rm}-y^*)}{c_{lm}(y^*, \mathbf{r})}\\
        & \le 1 + \frac{\alpha(|G|-2)}{1+\alpha(|G|-1)(\frac{2}{x_{rm}-x_{lm}}-1)} \le 2,
    \end{align*}
    where the last equality holds when $x_{rm}-x_{lm}=1$.

    Now we will show that this bound is tight. Consider a profile set with a single group where one agent is located at $0$ and all the other $|G|-1$ agents are located at $1$. The maximum cost of $y^*=1/2$ is $1/2 + \alpha(|G|-1)(1-1/2)$ while the maximal of the maximum cost is $1+\alpha(|G|-1)$, which locates the facility at $1$. One can verify that the approximation ratio in this profile set is $2$.
\end{proof}
\subsection*{Proof of Theorem~\ref{thm:5}}

    Left-point mechanism is strategyproof when both locations and group memberships are private, and has an approximation ratio of 3.

\begin{proof}
For the strategyproofness, we conclude that in each case, any misreporting eventually leads the facility location to remain unchanged or move farther away from every agent’s location, which leads their cost to remain unchanged or increase. Hence, each agent has no incentive to misreport.

Now we are ready to show the approximation ratio of Left-M. If there is only one group, from Lemma~\ref{lem:single group} we have the approximation ratio of $2$. If there are more than two groups, let the group which the leftmost agent belongs to be $G_{lm}$ and the maximum cost achieved by $y$ is incurred by the agent $i$ in group $G_r$. Let the leftmost agent be ${lm}$ and the optimal facility location be $y^*$. Hence, we have $y^* > y$. First, the maximum cost achieved by $y^*$ is at least
\begin{equation*}
    \mc(y^*,\mathbf{r}) \ge c_{lm}(y^*,\mathbf{r}) \ge (y^*-y).
\end{equation*}

Then from $y$ to $y^*$, the cost of agent $i$ will increase by the amount of at most $(y^*-y)+\alpha_r(|G_r|-1)(y^*-y)$. Hence, we have
\begin{align*}
    \mc(y,\mathbf{r}) &\le c_i(y^*,\mathbf{r}) + (y^*-y)+\alpha_r(|G_r|-1)(y^*-y) \\
    &\le \mc(y^*,\mathbf{r}) + (y^*-y)+\alpha_r(|G_r|-1)(y^*-y).
\end{align*}

Therefore, the approximation ratio is at most
\begin{align*}
    \rho &= \frac{\mc(y,\mathbf{r})}{\mc(y^*,\mathbf{r})} \le 1+\frac{(y^*-y)+\alpha_r(|G_r|-1)(y^*-y)}{y^*-y} \\
    &\le 3.
\end{align*}

\end{proof}
\subsection*{Proof of Theorem~\ref{thm:lof}}

    Last-leftmost or First-rightmost Mechanism is strategyproof and has an approximation ratio of $17/8$ when locations are private and group memberships are public.

\begin{proof}
    If $y=x_{lm}^*$, suppose that this group leftmost agent is in group $j^*$. For agents in group $j^*$, they can only move $x_{lm}^*$ to the left by misreporting, implying that the facility will not move to be closer to them. Hence, they have no incentive to misreport the location. For agents not in group $j^*$, they can only move their group leftmost agent location to the left or their group rightmost agent location to the right, which cannot move the facility. Hence, they have no incentive to misreport. If $y=x_{rm}^*$, we can also use a similar analysis to show the strategyproofness. Next, we will show the approximation ratio by case. Let $y^*$ be the optimal facility location, we will start from easier cases to the most difficult ones.

    \paragraph{Case 1: $y=x^*_{lm}$ and $y>y^*$.} \ \\
    \begin{center}
    \begin{tikzpicture}[scale = 2]
        \tikzstyle{every node}=[font=\footnotesize]
            \filldraw[black] (1/2,0) circle (0.5pt);
            \filldraw[black] (3/2,0) circle (0.5pt);
            \node at (0,-1/10) {$0$};
            \node at (1/2,-1/10) {$y^*$};
            \node at (3/2,-1/10) {$y(x^*_{lm})$};
            \node at (2,-1/10) {$1$};           
            \draw [-](0,0)--(2,0);
    \end{tikzpicture}
    \end{center}
    In this case, we can show that the maximum cost achieved by $y$ is incurred by a group leftmost agent. For the sake of contradiction, let the maximum cost achieved by $y$ be incurred by a group rightmost agent. Due to the definition of $x^*_{lm}$, each group rightmost agent is either at or on the right of $y$. Since $y^*$ is on the left of $y$, the cost of that group rightmost agent with respect to $y$ is at most of that with respect to $y^*$. However, $y^*$ is the optimal facility location, which makes a contradiction. Suppose that the maximum cost achieved by $y$ is incurred by $x^{j'}_{lm}$. Due to the definition of $x^*_{lm}$, we have $x^{j'}_{rm} \ge x^*_{lm}$. Hence, $y$ is within the interval $[x^{j'}_{lm}, x^{j'}_{rm}]$, and the maximum cost achieved by $y^*$ is at least the maximum cost among $x^{j'}_{lm}$ and $x^{j'}_{rm}$. From Lemma~\ref{lem:single group} we have the approximation ratio of $2$.

    \paragraph{Case 2: $y=x^*_{lm}$ and $y<y^*$.} \ \\
    \begin{center}
    \begin{tikzpicture}[scale = 2]
        \tikzstyle{every node}=[font=\footnotesize]
            \filldraw[black] (1/2,0) circle (0.5pt);
            \filldraw[black] (3/2,0) circle (0.5pt);
            \node at (0,-1/10) {$0$};
            \node at (3/2,-1/10) {$y^*$};
            \node at (1/2,-1/10) {$y(x^*_{lm})$};
            \node at (2,-1/10) {$1$};           
            \draw [-](0,0)--(2,0);
    \end{tikzpicture}
    \end{center}
    In this case, we can show that the maximum cost achieved by $y$ is incurred by a group rightmost agent. For the sake of contradiction, let the maximum cost achieved by $y$ be incurred by a group leftmost agent at $x^{j'}_{lm}$. Since $y^*$ is on the right of $y$, the cost of that agent with respect to $y$ is at most of that with respect to $y^*$. However, $y^*$ is the optimal facility location, which makes a contradiction. Suppose that the maximum cost achieved by $y$ is incurred by $x^{j'}_{rm}$. Due to the definition of $x^*_{lm}$, we have $x^{j'}_{lm} \le x^*_{lm}$. Hence, $y$ is within the interval $[x^{j'}_{lm}, x^{j'}_{rm}]$, and the maximum cost achieved by $y^*$ is at least the maximum cost among $x^{j'}_{lm}$ and $x^{j'}_{rm}$. From Lemma~\ref{lem:single group} we have the approximation ratio of $2$.

    \paragraph{Case 3: $y=x^*_{rm}$ and $y>y^*$.}\ \\
        \begin{center}
    \begin{tikzpicture}[scale = 2]
        \tikzstyle{every node}=[font=\footnotesize]
            \filldraw[black] (1/2,0) circle (0.5pt);
            \filldraw[black] (1,0) circle (0.5pt);
            \node at (0,-1/10) {$0$};
            \node at (1/2,-1/10) {$y^*$};
            \node at (1,-1/10) {$y(x^*_{rm})$};
            \node at (2,-1/10) {$1$};           
            \draw [-](0,0)--(2,0);
    \end{tikzpicture}
    \end{center}
    In this case, we can show that the maximum cost achieved by $y$ is incurred by a group leftmost agent. For the sake of contradiction, let the maximum cost achieved by $y$ be incurred by a group rightmost agent. Due to the definition of $x^*_{rm}$, each group rightmost agent is either at or on the right of $y$. Since $y^*$ is on the left of $y$, the cost of that group rightmost agent with respect to $y$ is at most of that with respect to $y^*$. However, $y^*$ is the optimal facility location, which makes a contradiction. Suppose that the maximum cost achieved by $y$ is incurred by $x^{j'}_{lm}$. Due to the definition of $x^*_{rm}$, we have $x^{j'}_{rm} \ge x^*_{rm}$. Hence, $y$ is within the interval $[x^{j'}_{lm}, x^{j'}_{rm}]$, and the maximum cost achieved by $y^*$ is at least the maximum cost among $x^{j'}_{lm}$ and $x^{j'}_{rm}$. From Lemma~\ref{lem:single group} we have the approximation ratio of $2$.

    \paragraph{Case 4: $y=x^*_{rm}$ and $y < y^*$.} 
    In this case, we can show that the maximum cost achieved by $y$ is incurred by a group rightmost agent. For the sake of contradiction, let the maximum cost achieved by $y$ be incurred by a group leftmost agent. Note that in this case, this agent can be on the right of $y$. Hence, we cannot use the previous analysis to deal with this case. If this group leftmost agent is on the left of $y$, we have the maximum cost achieved by $y$ is at most the maximum cost achieved by $y^*$ because of $y^*>y$, implying the approximation ratio of $1$. If this group leftmost agent is on the right of $y$, from the proof of Lemma~\ref{lem:single group} we know that the rightmost agent in the same group will incur a larger cost. Hence, the maximum cost achieved by $y$ can only be incurred by a group rightmost agent. Let this group be $G_{jr}$, the leftmost agent in $G_{jr}$ be $x_{lm}^{jr}$, and the rightmost agent in $G_{jr}$ be $x_{rm}^{jr}$ (for simplify the description, we will use $x_i$ to denote both agent $i$ and their location). If $y$ is at or on the right of $x_{lm}^{jr}$, from Lemma \ref{lem:single group} we have the approximation ratio of $2$. Now, we consider the case where $y$ is on the left of $x_{lm}^{jr}$.
    
    Let the group which the agent at $x^*_{rm}$ belongs to be group $G_{jl}$, the leftmost agent in $G_{jl}$ be $x_{lm}^{jl}$, and the rightmost agent in $G_{jl}$ be $x_{rm}^{jl}$. We first show that moving all agents in $G_{jl}$ to $y$ iteratively will not decrease the approximation ratio. To see this, the proof of Lemma~\ref{lem:single group} has already shown that the maximum cost can only be incurred by the leftmost agents and the rightmost agents in each group. Since $x_{lm}^{jl} \le x_{rm}^{jl} < y^*$. The maximum cost achieved by $y^*$ must not be incurred by $x_{rm}^{jl}$. Then we focus on agent $x_{lm}^{jl}$. Since agent $x_{rm}^{jl}$ is at $x^*_{rm}$, we do not need to move. We first move the second rightmost agent in $G_{jl}$ to $y$, which makes them closer to $y^*$. It will not decrease the maximum cost achieved by $y$ it is incurred by the $x_{rm}^{jr}$. Moreover, the maximum cost achieved by $y^*$ will not decrease since the cost of agent $x_{lm}^{jl}$ will decrease. Therefore, the approximation ratio will not decrease.
    Next, we will examine which agent incurs the maximum cost achieved by $y^*$. There are two possible cases, and we will analyze that the ratio in each of these cases does not exceed $17/8$.

    \textbf{Case 4-1.} $y^*$ is on the left of $x_{lm}^{jr}$. \\
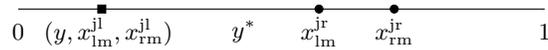
    \begin{center}
    \begin{tikzpicture}[scale = 2]
        \tikzstyle{every node}=[font=\footnotesize]

            \filldraw ([xshift=-2pt,yshift=-0.75pt]0.6,0) rectangle ++(1.5pt,1.5pt);
            \filldraw[black] (2,0) circle (0.75pt);
            \filldraw[black] (2.5,0) circle (0.75pt);
            \node at (0,-0.15) {$0$};
            \node at (0.6,-0.15) {$(y, x^{jl}_{lm}, x^{jl}_{rm})$};
            \node at (1.5,-0.15) {$y^*$};
            \node at (2,-0.15) {$x^{jr}_{lm}$};
            \node at (2.5,-0.15) {$x^{jr}_{rm}$};
            \node at (3.5,-0.15) {$1$};           
            \draw [-](0,0)--(3.5,0);
    \end{tikzpicture}
    \end{center}
    In this case, we focus on the impact of agent $x_{lm}^{jl}$ and agent $x_{rm}^{jr}$ on the maximum cost achieved by $y^*$. We will increase the approximation ratio by minimizing the maximum cost achieved by $y$. Since all agents in $G_{jl}$ are at the same location, the cost of agent $x_{lm}^{jl}$ is
        \begin{align*}
             c_{lm}^{jl}(y^*, \mathbf{r}) \!=\! d(y^*,x_{rm}^{jr})\!+\!\alpha_{jl}(|G_{jl}|\!-\!1)(1\!-\!d(y^*,x_{rm}^{jr}))\\
             = (1-\alpha_{jl}(|G_{jl}|-1))d(y^*,x_{rm}^{jr})+\alpha_{jl}(|G_{jl}|-1).
        \end{align*}

        Let $\alpha_{jl}(|G_{jl}|-1)$ be $c_{jl}\in [0,1]$. We further have
        \begin{align*}
            (1-c_{jl})d(y^*,x_{rm}^{jr})+c_{jl} \!-\! ((1-c'_{jl})d(y^*,x_{rm}^{jr})+c'_{jl}) \\
             = (1-d(y^*,x_{rm}^{jr}))(c_{jl}-c'_{jl}) \ge 0,
        \end{align*}
        where $c_{jl}>c'_{jl}$, implying that decreasing the value of $c_{jl}$ will not increase the maximum cost achieved by $y^*$ and not change the maximum cost achieved by $y$. Hence, the approximation ratio can only be larger by setting $c_{jl}=0$. Then we have $c_{lm}^{jl}(y^*, \mathbf{r}) = d(y^*,x_{rm}^{jr})$.
        Moreover, the cost of agent $x_{rm}^{jr}$ will increase the amount of at most $(y^*-y)-\alpha_{jr}(|G_{jr}|-1)(y^*-y)$ from $y^*$ to $y$ since all agents in $G_{jr}$ are on the right of $y^*$.Therefore, the approximation ratio is at most
        \begin{align*}
            \rho &\le \frac{c_{lm}^{jl}(y^*,\mathbf{r}) + (y^*-y)-\alpha_{jr}(|G_{jr}|-1)(y^*-y)}{c_{lm}^{jl}(y^*,\mathbf{r})}\\
            & \le \frac{2-\alpha_{jr}(|G_{jr}|-1)}{1} \le 2.
        \end{align*}
        
        \textbf{Case 4-2.} $y^*$ is on the right of $x_{lm}^{jr}$. \\
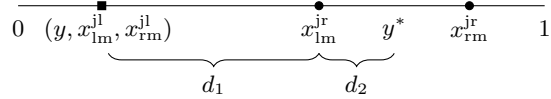
        \begin{center}
            \begin{tikzpicture}[scale = 2]
        \tikzstyle{every node}=[font=\footnotesize]

            \filldraw ([xshift=-2pt,yshift=-0.75pt]0.6,0) rectangle ++(1.5pt,1.5pt);
            \filldraw[black] (2,0) circle (0.75pt);
            \filldraw[black] (3,0) circle (0.75pt);
            \node at (0,-0.15) {$0$};
            \node at (0.6,-0.15) {$(y, x^{jl}_{lm}, x^{jl}_{rm})$};
            \node at (2,-0.15) {$x^{jr}_{lm}$};
            \node at (3,-0.15) {$x^{jr}_{rm}$};
            \node at (2.5,-0.15) {$y^*$};
            \node at (3.5,-0.15) {$1$};           
            \draw [-](0,0)--(3.5,0);

            \draw [decorate,decoration={brace,amplitude=5pt,mirror,raise=4ex}](0.6,0) -- (2,0) node[midway,yshift=-3em]{$d_1$};
            \draw [decorate,decoration={brace,amplitude=5pt,mirror,raise=4ex}](2,0) -- (2.5,0) node[midway,yshift=-3em]{$d_2$};
    \end{tikzpicture}
    \end{center}
        We consider the impact of agents $x_{lm}^{jl}$, $x_{lm}^{jr}$, and $x_{rm}^{jr}$ on the maximum cost achieved by $y^*$. First, we can use a similar analysis to show that setting $c_{jl}=0$ makes the approximation ratio larger. From Lemma~\ref{lem:single group} we know that the agents $x_{lm}^{jr}$ and $x_{rm}^{jr}$ will have the same cost with respect to the location $\frac{x_{lm}^{jr} + x_{rm}^{jr}}{2}$. 
        
        If the cost of agent $x_{lm}^{jl}$ is larger than the other two agents with respect to the location $\frac{x_{lm}^{jr} + x_{rm}^{jr}}{2}$. We can move $x_{lm}^{jr}$ to the left until it reaches $y$, or the cost of agent $x_{lm}^{jl}$ is equal to the other two agents with respect to the location $\frac{x_{lm}^{jr} + x_{rm}^{jr}}{2}$. After the movement, the approximation ratio will not decrease. To see this, one can see that the maximum cost achieved by $y^*$ is not incurred by $x_{lm}^{jr}$ before the movement is terminated. Moreover, the cost of $x_{rm}^{jr}$ with respect to $y^*$ will decrease. Hence, the maximum cost achieved by $y^*$ will not increase. Since $x_{lm}^{jr}$ will be closer to $y$, the maximum cost achieved by $y$ will not decrease. Therefore, the approximation ratio will not decrease. If the movement is terminated by reaching $y$, we can use Lemma~\ref{lem:single group} to show an approximation ratio $2$. 
        
        Next, we consider the case where the movement is terminated by the other condition. In this case, all three agents $x_{lm}^{jl}$, $x_{lm}^{jr}$, $x_{rm}^{jr}$ have the same cost with respect to the location $(x_{lm}^{jr} + x_{rm}^{jr})/2$. Then we focus on the approximation ratio between the maximum cost achieved by $y$ and $(x_{lm}^{jr} + x_{rm}^{jr})/2$ since the approximation ratio between the maximum cost achieved by $y$ and $y^*$ is less or equal to the former one. To simplify the description, let $d_1=x_{lm}^{jr}-x_{lm}^{jl}$ and $d_2 = (x_{lm}^{jr} + x_{rm}^{jr})/2$, and we just use $y^*$ for $(x_{lm}^{jr} + x_{rm}^{jr})/2$. First we have the maximum cost achieved by $y^*$ is at least
        \begin{equation*} 
            \mc(y^*,\mathbf{r}) \ge c_{lm}^{jl}(y^*,\mathbf{r}) = d_1+d_2.
        \end{equation*}
        Let $\alpha_{jr}(|G_{jr}|-1)$ be $c_{jr}$, we also have the maximum cost achieved by $y^*$ is at least
        \begin{equation*}
            \mc(y^*,\mathbf{r}) \ge c_{rm}^{jr}(y^*,\mathbf{r}) \ge d_2 + c_{jr}(1-d_2).
        \end{equation*}

        From $y^*$ to $x_{lm}^{jr}$, the cost of agent $x_{rm}^{jr}$ will increase by the amount of at most $(d_2+c_{jr}d_2)$, since there are at most $|G_{jr}|-1$ agents at $x_{lm}^{jr}$.  From $x_{lm}^{jr}$ to $y$, the cost of agent $x_{rm}^{jr}$ will increase by the amount of at most $(d_1-c_{jr}d_1)$, since all agents in $G_{jr}$ are at or on the right of $x_{lm}^{jr}$. Hence, the difference between the maximum cost achieved by $y$ and $y^*$ is at most
        \begin{equation*}
            \mc(y,\mathbf{r}) \le \mc(y^*,\mathbf{r}) + d_1 + d_2 + c_{jr}(d_2-d_1).
        \end{equation*}

        Then, we have the approximation ratio
        \begin{equation} \label{eq: rho*1}
            \rho = \frac{\mc(y,\mathbf{r})}{\mc(y^*,\mathbf{r})} \le 1+\frac{d_1 + d_2 + c_{jr}(d_2-d_1)}{d_1+d_2},
        \end{equation}
        and 
        \begin{equation} \label{eq: rho*2}
            \rho = \frac{\mc(y,\mathbf{r})}{\mc(y^*,\mathbf{r})} \le 1+\frac{d_1 + d_2 + c_{jr}(d_2-d_1)}{d_2 + c_{jr}(1-d_2)}.
        \end{equation}

        One can verify that Equation \ref{eq: rho*1} is monotonically decreasing with respect to $d_1$ and Equation \ref{eq: rho*2} is monotonically increasing with respect to $d_1$. The equality holds when $d_1=c_{jr}(1-d_2)$. By plugging $d_1$ into both equations, we have the approximation ratio
        \begin{align*}
            \rho & = \le 1+\frac{c_{jr}(1-d_2) + d_2 + c_{jr}(d_2-c_{jr}(1-d_2))}{d_2 + c_{jr}(1-d_2)}\\
            & =  1 + \frac{(1+c^2_{jr})d_2+c_{jr}-c^2_{jr}}{(1-c_{jr})d_2+c_{jr}},
        \end{align*}
        which is monotonically increasing with respect to $d_2$. Note that $d_1+2d_2\le 1$, we have $d_2\le \frac{1-c_{jr}}{2-c_{jr}}$. By plugging $d_2$ into the equation, we have the approximation ratio which is only related to $c_{jr}$. By carefully calculation, the approximation ratio reaches the maximum of $17/8$ when $c_{jr}=1/4$.
\end{proof}

\subsection*{Proof of Theorem~\ref{thm:mpm}}

    Middle-point Mechanism is strategyproof, and has an approximation ratio of $(29+20\sqrt{10})/54$  for minimizing the maximum cost when locations are public.

\begin{proof}
    Since Mid-M does not use the group memberships, it satisfies the strategyproofness trivially. Hence, we focus on the approximation ratio. Let the optimal facility location be $y^*$. Without loss of generality, we assume that $y< y^*$, the leftmost agent is in group $G_{l}$, and the maximum cost achieved by $y$ is incurred by an agent in $G_r$. To simplify the description, we use $x_{lm}^j$ and $x_{rm}^j$ to denote the leftmost and the rightmost agent in group $G_j$. First, we can see that the maximum cost achieved by $y$ is incurred by agent $x_{rm}^r$. 
    Next, we focus on the impact of agent $x_{lm}^{l}$, $x_{rm}^{r}$ on the maximum cost achieved by $y^*$.  First, we can move all agents in $G_l$ to $x_{lm}^{l}$, where the approximation ratio will not decrease. To see this, the maximum cost achieved by $y$ still be the same. The cost of agent $x_{lm}^{l}$ with respect to $y^*$ will decrease, implying that the maximum cost achieved by $y^*$ will not increase. To see this, for agents at or on the left of $y^*$, they move far away from $y^*$. For agents on the right of $y^*$, the distance from $y^*$ is less than the distance from $y^*$ to $x_{lm}^{l}$ since $y < y^*$ is the middle point. Hence, they also move far away from $y^*$.
    
    If the cost of agent $x_{lm}^{l}$ is larger than the other two agents with respect to the location $(x_{lm}^{r} + x_{rm}^{r})/2$. We can use a similar analysis with Theorem~\ref{thm:lof} to move $x_{lm}^{jr}$ to the left until it reaches $x_{lm}^{l}$, or the cost of agent $x_{lm}^{l}$ is equal to the other two agents with respect to the location $(x_{lm}^{r} + x_{rm}^{r})/2$. After the movement, the approximation ratio will not decrease. To see this, one can see that the maximum cost achieved by $y^*$ is not incurred by $x_{lm}^{r}$ before the movement is terminated. Moreover, the cost of $x_{rm}^{r}$ with respect to $y^*$ will decrease. Hence, the maximum cost achieved by $y^*$ will not increase. Since $x_{lm}^{r}$ will be closer to $y$, the maximum cost achieved by $y$ will not decrease. Therefore, the approximation ratio will not decrease. If the movement is terminated by reaching $y$, we have $y \ge (x_{lm}^{r} + x_{rm}^{r})/2$ while $y^* \le (x_{lm}^{r} + x_{rm}^{r})/2$. Hence, we have $y=y^*$ and the approximation ratio is $1$.

    Next, we consider the case where the movement is terminated by the other condition. In this case, all three agents $x_{lm}^{l}$, $x_{lm}^{r}$, $x_{rm}^{r}$ have the same cost with respect to the location $(x_{lm}^{r} + x_{rm}^{r})/2$. Then we focus on the approximation ratio between the maximum cost achieved by $y$ and $(x_{lm}^{r} + x_{rm}^{r})/2$ since the approximation ratio between the maximum cost achieved by $y$ and $y^*$ is less or equal to the former one. To simplify the description, let $d_1=y-x_{lm}^{l}$, $d_2 = (x_{lm}^{r} + x_{rm}^{r})/2$, and $d_3 = d(y, x_{lm}^{r})$. We just use $y^*$ for $(x_{lm}^{r} + x_{rm}^{r})/2$. 
    \paragraph{Subcase 1.} $y\le x_{lm}^{r}$
    First we have the maximum cost achieved by $y^*$ is at least
        \begin{equation*} 
            \mc(y^*,\mathbf{r}) \ge c_{lm}^{l}(y^*,\mathbf{r}) = d_1+d_2+d_3.
        \end{equation*}
        Let $\alpha_{r}(|G_{r}|-1)$ be $c_{r}$, we also have the maximum cost achieved by $y^*$ is at least
        \begin{equation*}
            \mc(y^*,\mathbf{r}) \ge c_{rm}^{r}(y^*,\mathbf{r}) \ge d_2 + c_{r}(1-d_2).
        \end{equation*}

        From $y^*$ to $x_{lm}^{r}$, the cost of agent $x_{rm}^{r}$ will increase by the amount of at most $(d_2+c_{r}d_2)$, since there are at most $|G_{r}|-1$ agents at $x_{lm}^{r}$.  From $x_{lm}^{r}$ to $y$, the cost of agent $x_{rm}^{jr}$ will increase by the amount of at most $(d_3-c_{r}d_3)$, since all agents in $G_{r}$ are at or on the right of $x_{lm}^{r}$. Hence, the difference between the maximum cost achieved by $y$ and $y^*$ is at most
        \begin{equation*}
            \mc(y,\mathbf{r}) \le \mc(y^*,\mathbf{r}) + d_3 + d_2  + c_{r}(d_2-d_3).
        \end{equation*}

        Then, we have the approximation ratio
        \begin{equation*}
            \rho = \frac{\mc(y,\mathbf{r})}{\mc(y^*,\mathbf{r})} \le 1+\frac{d_3 + d_2  + c_{r}(d_2-d_3)}{d_1+d_2+d_3}.
        \end{equation*}
        Since $y$ is the middle point, we have $d_1\ge d_3+2d_2$. Then we have 
        \begin{equation} \label{eq: rho*3f}
            \rho = \frac{\mc(y,\mathbf{r})}{\mc(y^*,\mathbf{r})} \le 1+\frac{d_3 + d_2  + c_{r}(d_2-d_3)}{3d_2+2d_3},
        \end{equation}
        Moreover, we have
        \begin{equation} \label{eq: rho*4f}
            \rho = \frac{\mc(y,\mathbf{r})}{\mc(y^*,\mathbf{r})} \le 1+\frac{d_3 + d_2  + c_{r}(d_2-d_3)}{d_2 + c_{r}(1-d_2)}.
        \end{equation}
        One can verify that Equation \ref{eq: rho*3f} is monotonically decreasing with respect to $d_3$ and Equation \ref{eq: rho*4f} is monotonically increasing with respect to $d_3$. The equality holds when $d_3=(c_r-c_rd_2-2d_2)/2$. By plugging $d_3$ into both equations, we have the approximation ratio
        \begin{align*}
            \rho & =1+\frac{c_{r}^2d_2 + c_{r}  + 3c_{r}d_2-c_{r}^2}{2d_2 + 2c_{r}(1-d_2)} \le 1+\frac{c_{r}^3 - 5c_{r}^2  + 5c_{r}}{2},
        \end{align*}
        which is monotonically increasing with respect to $d_2$. Note that $d_1+d_3+2d_2\le 1$ and $d_1\ge d_3+2d_2$, we have $d_2\le (1-c_r)/(2-c_r)$. By plugging $d_2$ into the equation, we have the approximation ratio which is only related to $c_{r}$. By carefully calculation, the approximation ratio reaches the maximum of $\frac{29+20\sqrt{10}}{54}$ when $c_{r}=\frac{5-\sqrt{10}}{3}$.

        \paragraph{Subcase 2.} $y> x_{lm}^{r}$
        First we have the maximum cost achieved by $y^*$ is at least
        \begin{equation*} 
            \mc(y^*,\mathbf{r}) \ge c_{lm}^{l}(y^*,\mathbf{r}) = d_1+d_2-d_3.
        \end{equation*}
        Let $\alpha_{r}(|G_{r}|-1)$ be $c_{r}$, we also have the maximum cost achieved by $y^*$ is at least
        \begin{equation*}
            \mc(y^*,\mathbf{r}) \ge c_{rm}^{r}(y^*,\mathbf{r}) \ge d_2 + c_{r}(1-d_2).
        \end{equation*}

        From $y^*$ to $y$, the cost of agent $x_{rm}^{r}$ will increase by the amount of at most $(d_2-d_3+c_{r}(d_2-d_3))$, since there are at most $|G_{r}|-1$ agents in $|G_{r}|$ on the left of $y$. Hence, the difference between the maximum cost achieved by $y$ and $y^*$ is at most
        \begin{equation*}
            \mc(y,\mathbf{r}) \le \mc(y^*,\mathbf{r}) + d_2 - d_3  + c_{r}(d_2-d_3).
        \end{equation*}

        Then, we have the approximation ratio
        \begin{equation*}
            \rho = \frac{\mc(y,\mathbf{r})}{\mc(y^*,\mathbf{r})} \le 1+\frac{d_2 -d_3  + c_{r}(d_2-d_3)}{d_1+d_2-d_3}.
        \end{equation*}
        Since $y$ is the middle point, we have $d_1\ge 2d_2-d_3$. Then we have 
        \begin{equation} \label{eq: rho*5}
            \rho = \frac{\mc(y,\mathbf{r})}{\mc(y^*,\mathbf{r})} \le 1+\frac{d_2 -d_3  + c_{r}(d_2-d_3)}{3d_2-2d_3} \le \frac{1+c_r}{3},
        \end{equation}
        where the equality holds when $d_3=0$. Since $d_1\ge 2d_2-d_3$ and $d_1-d_3+2d_2\le 2$, we have
        \begin{equation} \label{eq: rho*6}
        \begin{aligned}
            \rho &= \frac{\mc(y,\mathbf{r})}{\mc(y^*,\mathbf{r})} 
            \le 1+\frac{d_2 -d_3  + c_{r}(d_2-d_3)}{d_2 + c_{r}(1-d_2)} \\
            &\le 1+ \frac{(1+c_r)d_2}{d_2+c_r(1-d_2)} \le 1+\frac{1+c_r}{1+3c_r},
        \end{aligned}
        \end{equation}
        where the equality holds when $d_3=0$ and $d_2=1/4$. Combining Equations \ref{eq: rho*5} and \ref{eq: rho*6},  the approximation ratio reaches the maximum of $14/9$ when $c_{r}=2/3$.
\end{proof}

\subsection*{Proof of Theorem~\ref{thm:8}}

    Any deterministic mechanism has an approximation ratio of at least $2(1+2\sqrt{2})/7$ for minimizing the maximum cost when locations are public and group memberships are private.

\begin{proof}
    Given any strategyproof mechanism, consider a profile set $\mathbf{r}$ where all the agents' locations $x_1=\cdots=x_{\frac{n}{2}}=0$ and $x_{\frac{n}{2}+1}=\cdots=x_n=1$ ($n$ is an even number). Suppose that all agents at $0$ belong to $G_1$ and all agents at $1$ belong to $G_1$ where $\alpha_1=0$. In this profile set, the optimal facility location is $1/2$. Due to the symmetry, we assume that the mechanism will output $y\in [1/2,1]$. To achieve the approximation ratio less than ($2(1+2\sqrt{2})/7$), the mechanism should output $y\in [1/2, (1+2\sqrt{2})/7)$. 
    Now we consider another profile set $\mathbf{r}'$ where for all $\frac{n}{2}$ agents at $1$, $\alpha_2=(2-\sqrt{2})/(n/2-1)$. In this profile set, the optimal location is $\sqrt{2}/2$. To achieve the approximation ratio less than ($2(1+2\sqrt{2})/7$), the mechanism should output $y > (1+2\sqrt{2})/7$. Hence, agents at $1$ in $\mathbf{r}$ can benefit by misreporting their groups together, which makes the facility closer to them. From Lemma~\ref{lem:partial group SP} we have that mechanism is not strategyproof. 
\end{proof}

\section{Missing Lower Bounds}
\begin{theorem}
    Any deterministic mechanism  has an approximation ratio of at least $3/2$ for minimizing the social cost when locations are public and group memberships are private.
\end{theorem}
\begin{proof}
    Given any strategyproof mechanism, consider a profile set $\mathbf{r}$ where all the agents' locations $x_1=\cdots=x_{\frac{n}{2}}=0$ and $x_{\frac{n}{2}+1}=\cdots=x_n=1$ ($n$ is an even number). Suppose that all agents belong to $G_1$ where $\alpha_1=0$. Due to the symmetry, we assume that the mechanism will output $y\ge 1/2$.
    Now we consider another profile set $\mathbf{r}'$ where for all $\frac{n}{2}$ agents at $1$, $\alpha_2=1/(n/2-1)$. In this profile set, the optimal location is $0$. To achieve the approximation ratio less than $1.5$, the mechanism should output $y < 1/2$. Hence, agents at $1$ in $\mathbf{r}'$ can benefit by misreporting their groups together, which makes the facility from $0$ to $1$. From Lemma~\ref{lem:partial group SP} we have that mechanism is not strategyproof. 
\end{proof}

\clearpage
\end{document}